\documentclass[11pt, a4paper, noarxiv, onecolumn]{quantumarticle}

\usepackage{amsthm}

\usepackage{silence}
\theoremstyle{plain} 
\newtheorem{definition}{Definition}
\newtheorem{theorem}{Theorem}
\newtheorem{lemma}{Lemma}
\newtheorem{example}{Example}
\newtheorem{notation}{Notation}
\newtheorem{operator}{Operation}

\usepackage{graphicx}

\usepackage[most]{tcolorbox}
\usepackage{varwidth}
\usepackage{xspace}
\usepackage{complexity}
\usepackage[export]{adjustbox}
\tcbuselibrary{listings,breakable}
\tcbset{colback=cyan!5!white,colframe=cyan!75!black,fonttitle=\bfseries}
\undef\explainerbox

\tcbset{colback=black!5!white,colframe=white!75!black,fonttitle=\bfseries}

\renewcommand\implies{\Rightarrow}
\renewcommand\iff{=} 

\newtcolorbox{ruleblock}[2][]
{
    title    = {#2},
    top=-1em,
    skin first=enhanced,
    skin middle=enhanced,
    skin last=enhanced,
    #1,
}

\newcommand\concept[1]{\textit{#1}\xspace}

\usepackage{listings}
\PassOptionsToPackage{dvipsnames}{xcolor} 
\usepackage{xcolor}
\definecolor{LightGrey}{gray}{0.96}
\definecolor{CodeBlue}{rgb}{0.0,0.0,0.7}
\definecolor{CodeGreen}{rgb}{0.0,0.5,0.0}

\lstnewenvironment{pyconcode}{
  \lstset{
    language=Python,
    backgroundcolor=\color{LightGrey},
    frame=single,
    framerule=0.4pt,
    rulecolor=\color{black},
    basicstyle=\ttfamily\small,
    breaklines=true,
    showstringspaces=false,
    keepspaces=true,
    xleftmargin=1em,
    emph={from, import, print, lambda, for},
    emphstyle=\color{CodeGreen},
    emph={[2]wcnf_matrix, functools},
    emphstyle={[2]\color{CodeBlue}},
    literate={>>>}{{\textcolor{black}{>>>}}}3,
  }
}{}

\title{Quantum Physics using Weighted Model Counting}
\author{
  Dirck van den Ende, Joon Hyung Lee, Alfons Laarman, Henning Basold}
  \affiliation{  Leiden Institute of Advanced Computer Science}

\usepackage{commands}
\usepackage{format}
\usepackage{doi}

\begin{document}

\maketitle

\begin{abstract}
    \noindent
    Weighted model counting (WMC) has proven effective at a range of tasks within computer science, physics, and beyond. However, existing approaches for using WMC in quantum physics only target specific problem instances, lacking a general framework for expressing problems using WMC. This limits the reusability of these approaches in other applications and risks a lack of mathematical rigor on a per-instance basis. We present an approach for expressing linear algebraic problems, specifically those present in physics and quantum computing, as WMC instances. We do this by introducing a framework that converts Dirac notation to WMC problems. We build up this framework theoretically, using a type system and denotational semantics, and provide an implementation in Python. We demonstrate the effectiveness of our framework in calculating the partition functions of several physical models: The transverse-field Ising model (quantum) and the Potts model (classical). 
    The results suggest that heuristics developed in automated reasoning can be systematically applied to a wide class of problems in quantum physics through our framework.
\end{abstract}

\bigskip

\pagenumbering{gobble}
\thispagestyle{empty}
\tableofcontents
\thispagestyle{empty}
\pagenumbering{arabic}

\setcounter{page}{1}

\section{Introduction}
\label{sec:intro}

\let\oldpar\paragraph
\renewcommand\paragraph[1]{\oldpar{#1.}}

\paragraph{Problems in Quantum Physics}
Quantum physics describes the behavior of fundamental particles, atoms, and molecules. Quantum computing promises breakthroughs in areas such as drug development, traffic optimization, and artificial intelligence~\cite{Romero2025,Villanueva2025,Li2022}. However, large-scale quantum hardware remains scarce and error-prone, 
making classical simulation of quantum systems essential for validating 
algorithms, studying physical models, and benchmarking quantum devices. 
Yet classical simulation is inherently challenging due to the 
exponential growth of the solution space.  A central example is 
computing partition functions, which needs to sum over all possible 
configurations of a system and is crucial for understanding 
thermodynamic properties. Direct computation quickly becomes 
intractable as system size grows.

\paragraph{The Potential of Weighted Model Counting}

One promising classical technique to compute partition functions is weighted model counting (WMC).
Weighted model counting computes the total weight of all satisfying assignments to a \concept{weighted Boolean formula}.  Importantly, this framework extends to quantum physics: partition 
functions of quantum systems, such as the transverse-field Ising 
model, can be encoded as WMC instances, enabling classical solvers 
to directly tackle quantum problems.  We write:
\begin{equation}
    \WMC(\phi, W) = \sum_{\tau \models \phi} W(\tau),
\end{equation}
where $W$ assigns weights to assignments $\tau$ satisfying the Boolean formula $\phi$. This task is \#\P-hard~\cite{valiant1979complexity,Hunt1989}, yet modern solvers based on clause learning, algebraic decision diagrams, or tensor networks can handle large instances~\cite{Sang2005,Sharma2019,Dudek2020,Dudek2021,Dudek2021ProCount}.
Initial applications include probabilistic reasoning~\cite{Dilkas2021,Chavira2008,Paredes2019}.

These advances in clause learning and tensor network methods have enabled WMC applications in statistical physics and quantum computing:
Mei et al.~\cite{Mei2024b,Mei2024c,Mei2024} demonstrated that WMC can simulate quantum circuits by reducing gate operations to Boolean formulas. Even complex-valued simulations can be reduced to real or Boolean WMC instances. In parallel, Nagy et al.~\cite{Nagy2024} showed that the Ising model partition function can be reduced to the WMC problem, and that tensor-based WMC solvers outperform traditional physics tools like CATN~\cite{Pan2020}.

\paragraph{Contributions}

This paper makes the following contributions:
\begin{enumerate}
    \item A general encoding framework for $q^n \times q^m$ matrices using WMC, supporting operations such as addition, multiplication, and computing the trace. We define a formal language for encoding Dirac notation and prove correctness of the encoding.
    \item A Python implementation, \texttt{DiracWMC}, available at~\cite{DiracWMC}, used to generate the experiments in Section~\ref{sec:ising}, \ref{sec:quantumising}, and \ref{sec:potts}.
    \item Applications to the partition function of the transverse-field Ising model (quantum) and the Potts model (classical), demonstrating WMC beyond previously explored domains.
\end{enumerate}

Figure~\ref{fig:intro_figure} summarizes our framework pipeline: the user provides a physics problem and selects a WMC solver; the rest of the workflow is automated by our system.

\begin{figure}[h]
    \centering
    \includegraphics[width=\linewidth,page=1]{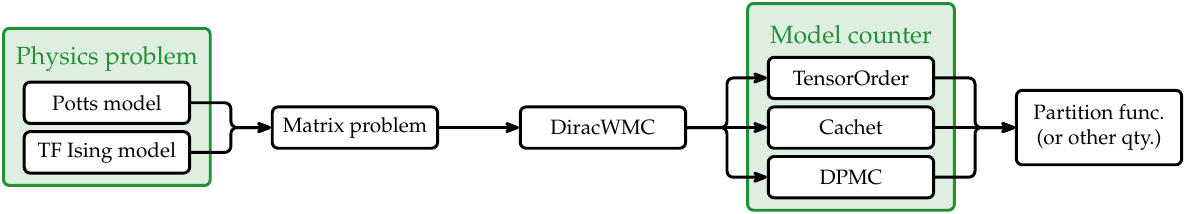}
    \caption{Workflow of solving physics problems using our framework.}
    \label{fig:intro_figure}
\end{figure}

\paragraph{Overview}

The matrices arising in partition function computations decompose as 
tensor products of small local operators, yielding matrices that are 
exponentially large but highly sparse. Direct matrix computation fails 
to exploit this sparsity; our approach instead encodes matrix entries 
as weighted Boolean formulas and reduces linear algebra to WMC, which 
is specifically designed to handle such structure efficiently. We define 
a formal language for Dirac notation, prove correctness of the encoding 
(Section~\ref{sec:matrix}), and implement the framework as 
\texttt{DiracWMC}. We demonstrate the approach on three models: the 
classical Ising model (Section~\ref{sec:ising}), the transverse-field 
Ising model (Section~\ref{sec:quantumising}), and the Potts model 
(Section~\ref{sec:potts}), with experimental comparisons against 
TensorOrder\if{and CATN}\fi included throughout. Preliminary 
definitions from Boolean logic, weighted model counting, and Dirac 
notation are collected in Section~\ref{sec:prelims}. We conclude in 
Section~\ref{sec:conclusion}.


\section{Preliminaries}\label{sec:prelims}

This section fixes notation used throughout. Unless stated otherwise, $\mathbb B=\{0,1\}$ and all matrices are over a field $\mathbb F$ (typically $\mathbb R$ or $\mathbb C$).

\subsection{Boolean logic and weighted model counting}\label{sec:boolean_logic}

Let $V$ be a finite set of Boolean variables, $\tau:V\to\mathbb B$ an assignment, and $\phi$ a Boolean formula over $V$. We write $\phi[\tau]\in\mathbb{B}$ for its truth value under~$\tau$. For $v\in V$, we define \concept{literals} $v$ and $\overline v$ (the negation $\neg v$ of $v$). A formula is in conjunctive normal form (CNF) if
\[
\phi \equiv \bigwedge_{i=1}^n \bigvee_{j=1}^m x_{ij},
\]
with each $x_{ij}$ a literal. All instances in our experiments are supplied to counters in CNF.

To assign weights to formulas, we use a weight function on assignments as introduced in Section~\ref{sec:intro}.
To enable powerful heuristics~\cite{Chavira2008}, the weight function is limited to literals as Definition~\ref{def:wmc} shows.
\begin{definition}[Weighted model counting]\label{def:wmc}
Let 
$W:V\times\mathbb B\to\mathbb F$ be a function called the \emph{weight function}. The \emph{weighted model count} of $\phi$ w.r.t.\ $W$ is
\begin{equation}
\WMC(\phi,W)\;=\;\sum_{\tau:V\to\mathbb B}\;\phi[\tau]\;\prod_{v\in V} W\bigl(v,\,\tau(v)\bigr).
\label{eq:wmc}
\end{equation}
We use $\dom(W)=V$, and abbreviate $W(v)=W(v,1)$ and $W(\overline v)=W(v,0)$.
\end{definition}

\subsection{Quantum notation}\label{sec:dirac}
We adopt standard Dirac notation and Kronecker algebra for $q^n\times q^m$ matrices.
\paragraph{Bras, kets, and matrix elements}
Moreover, we denote by $\bra{i}_q$ (row) and $\ket{i}_q$ (column) the computational basis vectors with a $1$ at position $i$ (zero-based) and $0$ elsewhere. For a matrix $M$, $\bra{j} M \ket{i} = M_{ij}$.
\paragraph{Kronecker product}
For matrices $A,B$, their Kronecker or tensor product is given by
\[
A\otimes B=\begin{bmatrix}
A_{0,0}B & \cdots & A_{0,r-1}B\\[-2pt]
\vdots & \ddots & \vdots\\[-2pt]
A_{c-1,0}B & \cdots & A_{c-1,r-1}B
\end{bmatrix},\qquad
(A_1\!\otimes\! B_1)(A_2\!\otimes\! B_2)=A_1A_2\!\otimes\! B_1B_2.
\]
\paragraph{Trace}
$\tr(M)$ denotes the matrix trace; $\tr(A\otimes B)=\tr(A)\tr(B)$.
\paragraph{Single-qubit gates}
As usual, we let
\[
X=\begin{bmatrix}0&1\\[2pt]1&0\end{bmatrix},\quad
Y=\begin{bmatrix}0&-i\\[2pt] i&0\end{bmatrix},\quad
Z=\begin{bmatrix}1&0\\[2pt]0&-1\end{bmatrix},\quad
H=\tfrac{1}{\sqrt2}\begin{bmatrix}1&1\\[2pt]1&-1\end{bmatrix}.
\]
All are involutions; we will use $X=HZH$ in Sec.~\ref{sec:quantumising}.
\paragraph{Matrix exponential}
For a square matrix $M$, $\exp(M)=\sum_{k\ge0} M^k/k!$; if $M$ is diagonal, $\exp(M)$ is the diagonal of entrywise exponentials. For large non-diagonal $M$ we use standard scaling–and–squaring with Padé approximants~\cite{Al-Mohy2010}.


\section{Matrix Computations using WMC}
\label{sec:matrix}

Here, we introduce weighted model counting (WMC) representations of general matrices. A WMC representation of quantum circuits was previously given in~\cite{Mei2024,Mei2024b,Mei2024c}. We aim to generalize and formalize this work by allowing an arbitrary base dimension of subspaces $q$ (qudits). We also add support for any $q^n \times q^m$ matrix, instead of just row/column vectors and square matrices.

We represent matrices as tuples $(\phi, W, x, y, q)$ of a Boolean formula, weight function, input and output variables, and some base size, respectively. The formula and weight function form the basis of model counting instances, used for every entry in the matrix. The input/output variables act as pointers to the specific entries in the matrix, which are obtained by adding restrictions to the values of these variables to the formula $\phi$. Scalars are represented by WMC instances $(\phi, W)$, where the value of the scalar is $\WMC(\phi, W)$.

We also provide encodings for several common matrix operations that can be performed on these representations directly, such as matrix multiplication, taking the trace, and computing the Kronecker product. To formalize the operations and prove their correctness, we first introduce a language of scalars and matrices, built from scalar constants, bras, and kets. This language is similar to D-Hammer, introduced by Xu et al.~\cite{Xu2025}. However, the language we introduce in this work is much simpler and neither supports contexts nor labeled matrices.

We introduce two kinds of denotational semantics on this language: $\denot{\cdot}_v$ returns the actual matrix or scalar that an expression represents. In contrast, $\denot{\cdot}_r$ returns an equivalence class of (matrix or scalar) representations that corresponds with the matrix or scalar.

Before we introduce this formal language, however, we give some code examples in our implementation of the language \texttt{DiracWMC}.

\subsection{Code examples}

The language is implemented using Python in the package \texttt{DiracWMC}~\cite{DiracWMC}. Full instructions on how to install and use the package are included there. Consider the following basic example, which calculates the product of a ket and a bra, and displays the result:
\begin{pyconcode}
>>> from wcnf_matrix import *
>>> I = Index(2)
>>> M = ket(I[0]) * bra(I[1])
>>> print(value(M))
[ 0.0  1.0
  0.0  0.0 ]
\end{pyconcode}
First, we import all of the contents of the \texttt{DiracWMC} package (called \texttt{wcnf\_matrix} in the Python code). Then we create a space in which to perform operations, which we do using \texttt{Index}. The number $2$ indicates that we use $q = 2$, i.e., we are working with qubits. The third line then creates two objects, a ket and a bra. These do not store the values of the two vectors explicitly, but rather store the tuples $(\phi, W, x, y, q)$ that can be used to calculate the entries in the matrices. In this example, we multiply the column vector $(1, 0)^T$ with the row vector $(1, 0)$. This then results in a new object, which again stores a tuple $(\phi, W, x, y, q)$ instead of the entries in the matrix. To get the actual entries of the matrix, we use \texttt{value(M)}.

We can replace the index number $2$ with $3$ to get a $3 \times 3$ matrix instead:
\begin{pyconcode}
>>> from wcnf_matrix import *
>>> I = Index(3)
>>> M = ket(I[0]) * bra(I[1])
>>> print(value(M))
[ 0.0  1.0  0.0
  0.0  0.0  0.0
  0.0  0.0  0.0 ]
\end{pyconcode}
Matrices can be multiplied and added, and the Kronecker product of two matrices can be determined:
\begin{pyconcode}
>>> from wcnf_matrix import *
>>> I = Index(3)
>>> M1 = ket(I[0]) * bra(I[1])
>>> M2 = ket(I[2]) * bra(I[0])
>>> print(value(M1 * M2))
[ 0.0  0.0  0.0
  0.0  0.0  0.0
  0.0  0.0  0.0 ]
>>> print(value(3.3 * M1 + M2))
[ 0.0  3.3  0.0
  0.0  0.0  0.0
  1.0  0.0  0.0 ]
>>> print(value(M1 ** M2))
[ 0.0  0.0  0.0  0.0  0.0  0.0  0.0  0.0  0.0
  0.0  0.0  0.0  0.0  0.0  0.0  0.0  0.0  0.0
  0.0  0.0  0.0  0.0  0.0  0.0  0.0  0.0  0.0
  0.0  0.0  0.0  0.0  0.0  0.0  0.0  0.0  0.0
  0.0  0.0  0.0  0.0  0.0  0.0  0.0  0.0  0.0
  0.0  0.0  0.0  0.0  0.0  0.0  0.0  0.0  0.0
  0.0  1.0  0.0  0.0  0.0  0.0  0.0  0.0  0.0
  0.0  0.0  0.0  0.0  0.0  0.0  0.0  0.0  0.0
  0.0  0.0  0.0  0.0  0.0  0.0  0.0  0.0  0.0 ]
\end{pyconcode}
The true power of the package lies in cases where the explicit values of the matrix are not required. For example, the trace can be calculated using \texttt{M.trace\_formula()}. This returns \texttt{CNF} and \texttt{WeightFunction} objects. The trace can be calculated by passing the \texttt{CNF} formula as an argument to a call to the \texttt{WeightFunction} object.
\begin{pyconcode}
>>> from wcnf_matrix import *
>>> from functools import reduce
>>> I = Index(2)
>>> M1 = 2 * ket(I[0]) * bra(I[0]) + ket(I[1]) * bra(I[1])
>>> print(value(M1))
[ 2.0  0.0
  0.0  1.0 ]
>>> M2 = reduce(lambda x, y: x ** y, [M1]*100)
>>> cnf, weight_func = M2.trace_formula()
>>> print(weight_func(cnf))
5.15378e+47
\end{pyconcode}
In this example, we use the Python built-in \texttt{reduce} to create an object \texttt{M2} that represents a $2^{100} \times 2^{100}$ matrix. Then we calculate the trace of this large matrix using a model counter.

We could also retrieve entries in this matrix by multiplying it by bras and kets. In the following example, we calculate the top-left entry in the matrix:
\begin{pyconcode}
... (continued) ...
>>> B = reduce(lambda x, y: x ** y, [bra(I[0])]*100)
>>> K = reduce(lambda x, y: x ** y, [ket(I[0])]*100)
>>> print(value(B * M2 * K))
[ 1.26765e+30 ]
\end{pyconcode}
It is also possible to label the different dimensions of the matrix, such that matrices acting on different subspaces can be multiplied and added. In the following example, we apply a matrix \texttt{M} on a subspace labeled with \texttt{r1} (using the syntax \texttt{M | r1}), while the vector we apply it to acts on subspaces \texttt{r1} and \texttt{r2}. This example also shows the use of \texttt{uset}, which returns an iterable \texttt{I[0]}, \texttt{I[1]}, \dots, \texttt{I[q-1]}.
\begin{pyconcode}
>>> from wcnf_matrix import *
>>> from functools import reduce
>>> I = Index(2)
>>> r1, r2 = Reg(I), Reg(I)
>>> phi = lambda index: reduce(lambda x, y: x + y, (ket(nv, nv) for nv in
... uset(index))) # Returns column vector (1, 1, 1, 1)^T
>>> M = ket(I[0]) * bra(I[0]) # Matrix [(1, 0), (0, 0)]
>>> print(value((M | r1) * (phi(I) | (r1, r2))))
[ 1.0
  0.0
  0.0
  0.0 ] | (reg0, reg1)
\end{pyconcode}
Note that the resulting matrix is still labeled. To get a matrix \texttt{M} without labels, use \texttt{M.mat}.

By default, the package uses the DPMC model counter. However, it is possible to change this to Cachet or TensorOrder using the \texttt{set\_model\_counter} method. It is also possible to set the type of variable encoding used in the matrix representations, which may have a performance impact for $q > 2$. Variable encodings are discussed in more detail later in this section and in Appendix~\ref{chapter:encodings}.
\begin{pyconcode}
>>> from wcnf_matrix import *
>>> set_model_counter(DPMC) # Default
>>> set_model_counter(Cachet)
>>> set_model_counter(TensorOrder)
>>> set_var_rep_type(LogVarRep) # Default
>>> set_var_rep_type(OrderVarRep)
>>> set_var_rep_type(OneHotVarRep)
\end{pyconcode}

\subsection{Language Syntax}
\label{sec:syntax}

The formal language we introduce has two types of expressions: scalars and matrices. Scalars from a field $\mathbb F$ are of type $\mathcal S$. Matrices have a type that contains the size of the matrix and its base size. A base size of $q \in \mathbb Z_{\geq 2}$ is used to represent $q^n \times q^m$ matrices. The intuition of this number is that it represents the dimension of the smallest vector space that all of our matrices act on. For a system of qubits, for example, a base size of $q = 2$ would be used, since elementary operations on qubits are performed using $2 \times 2$ unitary matrices. Any unitary acting on multiple qubits has dimensions that are a power of two.

We write the type of a matrix as $\mathcal M(q, m \to n)$, representing a $q^n \times q^m$ matrix, with $n,m \in \mathbb Z_{\geq0}$. Note that $m$ and $n$ do not indicate the size of the matrix directly, but rather the number of ``input and output subspaces''. Also note the reversal of the order of $n$ and $m$. We use this notation because a $q^n \times q^m$ matrix ($q^n$ rows and $q^m$ columns) is generally interpreted as a linear map $\mathbb F^{q^m} \to \mathbb F^{q^n}$.

More formally, for $n, m \in \mathbb Z_{\geq0}$ and $q \in \mathbb Z_{\geq2}$ the type syntax is
\begin{align}
    T ::= \mathcal S \mid \mathcal M(q, m \to n)
\end{align}
The syntax of expressions $e$ is split up into scalars $s$ and matrices $M$.
\begin{align}
    e &::= s \mid M \\
    s &::= \alpha \mid s_1 \cdot s_2 \mid s_1 + s_2 \mid \mathsf{tr}(M) \mid \mathsf{entry}(i,j,M) \mid \mathsf{apply}(f,s) \\
    M &::= \mathsf{bra}(i,q) \mid \mathsf{ket}(i,q) \mid M_2 \cdot M_1 \mid M_1 + M_2 \mid M_1 \otimes M_2 \\
    &\quad\ \ \mid s \cdot M \mid \mathsf{trans}(M) \mid \mathsf{apply}(f,M)
\end{align}
Here $\alpha \in \mathbb F$ is an arbitrary constant and $f: \mathbb F \to \mathbb F$ is an arbitrary field endomorphism. 

Scalar expressions can be combined using multiplication and addition. Applying a field endomorphism to a scalar also results in another scalar. In addition, taking the trace or getting a specific entry from a matrix gives a scalar.

The most basic matrices are the bra and ket, which are expressions for length-$q$ row and column computational basis vectors, respectively. These vectors have zeros everywhere except at the entry with index $0 \leq i < q$. Matrices can also be multiplied and added. In addition, we have syntax for taking the Kronecker product, matrix-scalar multiplication, and taking the transpose of a matrix. We also add support for applying a field endomorphism $f$ to every entry of the matrix.

\subsection{Type system}
\label{sec:type_system}

\setlength{\jot}{12pt}

The type system associates expressions with types. We say that the expression $e$ has type $T$ if $\vdash e : T$ can be proven using the type rules below.

\subsubsection{Scalar type rules}

The usual rules for scalars apply: Multiplying or adding two scalars results in a scalar, and applying a field endomorphism to a scalar yields a scalar as well. In addition, any element of $\mathbb F$ is a scalar.
\begin{gather}
    \begin{prooftree}
        \hypo{\alpha \in \mathbb F}
        \infer1[(Const)]{\vdash \alpha : \mathcal S}
    \end{prooftree} \quad \begin{prooftree}
        \hypo{\vdash s_1 : \mathcal S}
        \hypo{\vdash s_2 : \mathcal S}
        \infer2[(Mul)]{\vdash s_1 \cdot s_2  : \mathcal S}
    \end{prooftree} \quad \begin{prooftree}
        \hypo{\vdash s_1 : \mathcal S}
        \hypo{\vdash s_2 : \mathcal S}
        \infer2[(Add)]{\vdash s_1 + s_2  : \mathcal S}
    \end{prooftree} \\
    \begin{prooftree}
        \hypo{\vdash s : \mathcal S}
        \hypo{f: \mathbb F \to \mathbb F\text{ is a field endomorphism}}
        \infer2[(Apply)]{\vdash \mathsf{apply}(f,s)  : \mathcal S}
    \end{prooftree}
\end{gather}
Getting an entry from a matrix or calculating the trace of a square matrix also results in a scalar:
\begin{gather}
    \begin{prooftree}
        \hypo{\vdash M : \mathcal M(q, m \to n)}
        \infer1[(Entry)]{\vdash \mathsf{entry}(i,j,M) : \mathcal S}
    \end{prooftree} \qquad \begin{prooftree}
        \hypo{\vdash M : \mathcal M(q, n \to n)}
        \infer1[(Trace)]{\vdash \mathsf{tr}(M) : \mathcal S}
    \end{prooftree}
\end{gather}
where $0 \leq i < q^n$ and $0 \leq j < q^m$.

\subsubsection{Matrix type rules}

The bra and ket form the basis for matrix expressions. These have the types $\mathcal M(q, 1 \to 0)$ and $\mathcal M(0 \to 1)$ respectively, as they can be interpreted as linear maps $\mathbb F^q \to \mathbb F$ and $\mathbb F \to \mathbb F^q$. For $0 \leq i < q$ we have
\begin{align}
    \begin{prooftree}
        \infer0[(Bra)]{\vdash \mathsf{bra}(i,q) : \mathcal M(q, 1 \to 0)}
    \end{prooftree} \quad \begin{prooftree}
        \infer0[(Ket)]{\vdash \mathsf{ket}(i,q) : \mathcal M(q, 0 \to 1)}
    \end{prooftree}
\end{align}
Matrix multiplication is essentially the composition of maps $\mathbb F^{q^m} \to \mathbb F^{q^k}$ and $\mathbb F^{q^k} \to \mathbb \mathbb F^{q^n}$ to one map $\mathbb F^{q^m} \to \mathbb F^{q^n}$. However, do note that the composition is read from right to left. Hence $M_2$ and $M_1$ are swapped. 
\begin{align}
    \begin{prooftree}
        \hypo{\vdash M_1 : \mathcal M(q, m \to k)}
        \hypo{\vdash M_2 : \mathcal M(q, k \to n)}
        \infer2[(MatMul)]{\vdash M_2 \cdot M_1 : \mathcal M(q, m \to n)}
    \end{prooftree}
\end{align}
Adding two matrices of the same type results in a matrix with that type. Multiplying a matrix by a scalar results in a matrix of the same type, and so does applying a field endomorphism entry-wise.
\begin{gather}
    \begin{prooftree}
        \hypo{\vdash M_1 : \mathcal M(q, m \to n)}
        \hypo{\vdash M_2 : \mathcal M(q, m \to n)}
        \infer2[(MatAdd)]{\vdash M_1 + M_2 : \mathcal M(q, m \to n)}
    \end{prooftree} \\ \begin{prooftree}
        \hypo{\vdash s : \mathcal S}
        \hypo{\vdash M : \mathcal M(q, m \to n)}
        \infer2[(ScaMul)]{\vdash s \cdot M : \mathcal M(q, m \to n)}
    \end{prooftree} \\ \begin{prooftree}
        \hypo{\vdash M : \mathcal M(q, m \to n)}
        \hypo{f: \mathbb F \to \mathbb F\text{ is a field endomorphism}}
        \infer2[(MatApply)]{\vdash \mathsf{apply}(f,M) : \mathcal M(q, m \to n)}
    \end{prooftree}
\end{gather}
Taking the transpose of an $q^n \times q^m$ matrix results in a $q^m \times q^n$ matrix:
\begin{gather}
    \begin{prooftree}
        \hypo{\vdash M : \mathcal M(q, m \to n)}
        \infer1[(Trans)]{\vdash \mathsf{trans}(M) : \mathcal M(q, n \to m)}
    \end{prooftree}
\end{gather}
The Kronecker product of a $q^{n_1} \times q^{m_1}$ matrix and a $q^{n_2} \times q^{m_2}$ is a $q^{n_1+n_2} \times q^{m_1+m_2}$ matrix:
\begin{align}
    \begin{prooftree}
        \hypo{\vdash M_1 : \mathcal M(\mathbb F, q, m_1 \to n_1)}
        \hypo{\vdash M_2 : \mathcal M(\mathbb F, q, m_2 \to n_2)}
        \infer2[(Kron)]{\vdash M_1 \otimes M_2 : \mathcal M(q, m_1 + m_2 \to n_1 + n_2)}
    \end{prooftree}
\end{align}

\if{\begin{example}\label{ex:type_system}
    As an example we prove that $(3 \cdot \mathsf{ket}(0,2) \cdot \mathsf{bra}(1,2)) \otimes \mathsf{ket}(0,2)$ has type $\mathcal M(2, 1 \to 2)$, where we have $q = 2$. We use the field of complex numbers $\mathbb F = \mathbb C$.
    First we prove that $3 \cdot \mathsf{ket}(0,2) \cdot \mathsf{bra}(1,2)$ is of type $\mathcal M(2, 1 \to 1)$:
    \begin{align}
        \begin{prooftree}
                    \hypo{3 \in \mathbb C}
                \infer1{\vdash 3 : \mathcal S}
                    \infer0{\vdash \mathsf{bra}(1,2) : \mathcal M(2, 1 \to 0)}
                    \infer0{\vdash \mathsf{ket}(0,2) : \mathcal M(2, 0 \to 1)}
                \infer2{\vdash \mathsf{ket}(0,2) \cdot \mathsf{bra}(1,2) : \mathcal M(2, 1 \to 1)}
            \infer2{\vdash 3 \cdot \mathsf{ket}(0,2) \cdot \mathsf{bra}(1,2) : \mathcal M(2, 1 \to 1)}
        \end{prooftree}
    \end{align}
    Applying the Kronecker product and using the type of $\mathsf{ket}(0,2)$ gives
    \begin{align}
        \begin{prooftree}
                \hypo{\vdash 3 \cdot \mathsf{ket}(0,2) \cdot \mathsf{bra}(1,2) : \mathcal M(2, 1 \to 1)}
                \infer0{\vdash \mathsf{ket}(0,2) : \mathcal M(2, 0 \to 1)}
            \infer2{\vdash (3 \cdot \mathsf{ket}(0,2) \cdot \mathsf{bra}(1,2)) \otimes \mathsf{ket}(0,2) : \mathcal M(2, 1 \to 2)}
        \end{prooftree}
    \end{align}
\end{example}}\fi

\subsection{Value denotational semantics}
\label{sec:value_semantics}

As a baseline for the representation semantics we introduce later, we define the denotational semantics $\denot{\cdot}_v$ as the ``value'' of an expression, i.e., the concrete scalar or matrix that the expression represents. For example, the value of $\mathsf{ket}(1,2) \cdot \mathsf{bra}(0, 2)$ is a $2 \times 2$ matrix with a $1$ in the bottom-left corner and $0$ everywhere else. We define the denotational semantics inductively on the type derivations of the expression, meaning expressions without a type (e.g. $\mathsf{ket}(1,2) \cdot \mathsf{ket}(0,2)$) do not get a value. Note that, due to the way the type system is defined, there is at most one type derivation for each expression.

For scalars, the denotational semantics are defined as follows:
\begin{gather}
    \begin{aligned}
        \denot{\alpha}_v &= \alpha \\
        \denot{s_1 + s_2}_v &= \denot{s_1}_v + \denot{s_2}_v \\
        \denot{s_1 \cdot s_2}_v &= \denot{s_1}_v \cdot \denot{s_2}_v
    \end{aligned}\qquad \begin{aligned}
        \denot{\mathsf{apply}(f,s)}_v &= f(\denot{s}_v) \\
        \denot{\mathsf{entry}(i,j,M)}_v &= (\denot{M}_v)_{ij} \\
        \denot{\mathsf{tr}(M)}_v &= \tr(\denot{M}_v)
    \end{aligned}
\end{gather}
The semantics of bras and kets are the $1 \times q$ and $q \times 1$ matrices with an entry $1$ at the $i$-th position (counting from zero) and $0$ everywhere else. We denote these matrices using Dirac notation with $\bra{i}_q$ and $\ket{i}_q$, where the $q$ is left out if it is clear from context.
\begin{align}
    \denot{\mathsf{bra}(i,q)}_v = \bra{i}_q \qquad \denot{\mathsf{ket}(i,q)}_v = \ket{i}_q
\end{align}
The semantics of other matrix operations are performed simply by evaluating the expression recursively:
\begin{align}
    \begin{aligned}
        \denot{M_2 \cdot M_1}_v &= \denot{M_2}_v \cdot \denot{M_1}_v \\
        \denot{M_1 + M_2}_v &= \denot{M_1}_v + \denot{M_2}_v \\
        \denot{M_1 \otimes M_2}_v &= \denot{M_1}_v \otimes \denot{M_2}_v \\
    \end{aligned}\qquad \begin{aligned}
        \denot{s \cdot M}_v &= \denot{s}_v \cdot \denot{M}_v \\
        \denot{\mathsf{trans}(M)}_v &= \denot{M}_v^T \\
        \denot{\mathsf{apply}(f,M)}_v &= f(\denot{M}_v)
    \end{aligned}
\end{align}
Note that the type rules above prohibit any incompatible matrices from being multiplied or added. The value of $f(M)$ for a field endomorphism $f: \mathbb F \to \mathbb F$ and matrix $M$ is the matrix $M$ with $f$ applied to every entry.

These denotational semantics are in a certain sense valid, because $\denot{e}_v \in \denot{T}_v$ for any expression $e$ of type $T$.

\begin{example}{}{}
    The value semantics of the expression
    \begin{align}
        e = (3 \cdot \mathsf{ket}(0, 2) \cdot \mathsf{bra}(1, 2)) \otimes \mathsf{ket}(0, 2)
    \end{align} 
    can be determined as 
    \[\denot{e}_v = \begin{bmatrix}
            0 & 3 \\
            0 & 0 \\
            0 & 0 \\
            0 & 0
        \end{bmatrix}.\]
    \if{\begin{align}
        \denot{e}_v
        &= \denot{3 \cdot \mathsf{ket}(0,2) \cdot \mathsf{bra}(1,2)}_v \otimes \denot{\mathsf{ket}(0,2)}_v \\
        &= (\denot{3}_v \cdot \denot{\mathsf{ket}(0,2) \cdot \mathsf{bra}(1,2)}_v) \otimes \denot{\mathsf{ket}(0,2)}_v \\
        &= (3 \cdot \denot{\mathsf{ket}(0,2) \cdot \mathsf{bra}(1,2)}_v) \otimes \denot{\mathsf{ket}(0,2)}_v \\
        &= (3 \cdot \denot{\mathsf{ket}(0,2)}_v \cdot \denot{\mathsf{bra}(1,2)}_v) \otimes \denot{\mathsf{ket}(0,2)}_v \\
        &= (3 \cdot \ket{0}_2 \cdot \bra{1}_2) \otimes \ket{0}_2 \\
        &= \left(3 \cdot \begin{bmatrix}
            1 \\ 0
        \end{bmatrix} \cdot \begin{bmatrix}
            0 & 1
        \end{bmatrix}\right) \otimes \begin{bmatrix}
            1 \\ 0
        \end{bmatrix} \\
        &= \begin{bmatrix}
            0 & 3 \\
            0 & 0 \\
            0 & 0 \\
            0 & 0
        \end{bmatrix}
    \end{align}
    The expression is associated with the value we would like it to have.}\fi
\end{example}

\subsection{Representations}
\label{sec:repr}

We introduce representations for both scalars and matrices. Scalars will have a representation that is the solution to a model counting problem $(\phi, W)$ (i.e., $\WMC(\phi, W)$). Meanwhile, matrices are represented with a longer tuple that also includes input and output variables, and a base size $q$: $(\phi, W, x, y, q)$.

\subsubsection{Scalar representation}
\label{sec:scalar_repr}

As mentioned above, scalars are represented by model counting instances $(\phi, W)$. This is formalized in the following definition:

\begin{definition}[Scalar representation]\label{def:scalar_rep}
    A tuple $(\phi, W)$ of a Boolean formula $\phi$ over a set of variables $V$ and a weight function $W: V \times \mathbb B \to \mathbb F$ \textbf{represents} an element $\alpha \in \mathbb F$ if $\WMC(\phi, W) = \alpha$.
    We write in this case
    \begin{align}
        \rep(\phi, W) = \alpha.
    \end{align}
\end{definition}

\begin{example}{}{}
    Suppose we have a Boolean formula $\phi$ and weight function $W: \{x, y\} \times \mathbb B \to \mathbb F$ given by
    \begin{align}
        \SwapAboveDisplaySkip
        \phi &\equiv x \to y \\
        W(x) &= W(\overline x) = 1 \\
        W(y) &= W(\overline y) = 1/2
    \end{align}
    Then $\rep(\phi, W) = \WMC(\phi, W) = 3/2$.
\end{example}

\subsubsection{Matrix representation}
\label{sec:matrix_repr}

The definition of a matrix representation extends on that of a scalar representation by adding input and output variables $x$ and $y$, and a base size $q$. The input and output variables serve as pointers to the different entries in the matrix. The formula $\phi$ and weight function $W$ are used as the basis for a set of model counting problems, one for every entry in the matrix. The same weight function is used at every entry, but the formula $\phi$ is extended with requirements for the input and output variables: $\phi' \equiv \phi \land (x = j) \land (y = i)$. The value at the entry is the weighted model count $\WMC(\phi', W)$.

The input and output of a matrix representation consist of Boolean variables that together represent some number in the range $\{0, \dots, q^n-1\}$. We do this by using strings (of length $n$) of ``$q$-state variable encodings''. These encodings use Boolean variables and formulae to represent numbers from $\{0, \dots, q-1\}$. How these can be implemented is described in Appendix~\ref{chapter:encodings}. However, there are some important properties these encodings need to have. These are outlined below:
\begin{itemize}
    \item For an encoding $v$ we write $\var(v)$ for the set of all Boolean variables $v$ uses.
    \item For an encoding $v$ we can write $v = n$ to indicate $v$ is equal to some number $n$. There should be exactly one assignment $\tau: \var(v) \to \mathbb B$ for which $(v = n)[\tau] = 1$.
    \item Denote $\val_v \equiv \bigvee_{n=0}^{q-1} (v=n)$.
    \item Write $v \leftrightarrow w$ for the equality of two $q$-state encodings $v$ and $w$.
\end{itemize}
We use the same notation for strings of these variable encodings.

\begin{definition}[Matrix representation]\label{def:matrix_rep}
    Suppose we have a tuple $(\phi, W, x, y, q)$ of a Boolean formula $\phi$ over a set of variables $V$, a weight function $W: V \times \mathbb B \to \mathbb F$, two strings of $q$-state variables $x$ and $y$ over $V$, and a base size $q \in \mathbb Z_{\geq 2}$. This tuple \textbf{represents the matrix} $M \in \mathrm{Mat}(\mathbb F, q^{|y|} \times q^{|x|})$ if for all $j \in \{0, \dots, q^{|x|} - 1\}$ and $i \in \{0, \dots, q^{|y|} - 1\}$ we have
    \begin{align}
        \bra{j}M\ket{i} = M_{ij} &= \WMC\left(\phi \land x = j \land y = i, W\right)
        \label{eq:matrix_rep}
    \end{align}
    Every tuple represents exactly one matrix, which justifies the notation
    \begin{align}
        \rep(\phi,W,x,y,q) = M.
    \end{align}
\end{definition}

\if{\begin{example}{}{}
    We give a representation of the matrix
    \begin{align}
        M = \begin{bmatrix}
            0 & 2 \\
            1 & 0
        \end{bmatrix}
    \end{align}
    with base size $q = 2$. This means we can use Boolean variables as our $q$-state input and output variables $x$ and $y$.
    Note that any non-zero entry $(i, j)$ in the matrix has $i \neq j$, which means we use the formula $\phi \equiv x \leftrightarrow \overline y$. We need the model count to be $2$ when $x$ is true and $y$ is false, which is why we set $W(x) = 2$. Any other value of $W$, we set to $1$. When determining the weighted model count for every entry of the matrix, we find that $(\phi,W,x,y,2)$ is a representation of $M$:
    \begin{align}
        \begin{aligned}
            \WMC((x \leftrightarrow \overline y) \land \overline x \land \overline y,W) &= 0 \\
            \WMC((x \leftrightarrow \overline y) \land \overline x \land y,W) &= W(\overline x)W(y) = 1 \\
            \WMC((x \leftrightarrow \overline y) \land x \land \overline y,W) &= W(x)W(\overline y) = 2 \\
            \WMC((x \leftrightarrow \overline y) \land x \land y,W) &= 0
        \end{aligned}
    \end{align}
    The first and last equations are zero since the formulae are unsatisfiable.
\end{example}}\fi

\if{\begin{example}{}{}
    There is no requirement that the input and output variables be different. Such can be the case for diagonal matrices, and the Pauli-$Z$ matrix in particular:
    \begin{align}
        Z = \begin{bmatrix}
            1 & 0 \\
            0 & -1
        \end{bmatrix}
    \end{align}
    This matrix can represented with $(\top, W, x, x, 2)$, using the weight function $W: \{x\} \times \mathbb B \to \mathbb F$ defined by $W(\overline x) = 1$ and $W(x) = -1$. 
\end{example}}\fi

\subsubsection{Representation map}
\label{sec:repr_map}

From Definitions~\ref{def:scalar_rep} and \ref{def:matrix_rep} we introduce the map
\begin{align}
    \rep: \Rep \to \mathbb F \cup \Mat(\mathbb F),    
\end{align}
where $\Rep$ is the set of scalar and matrix representations. This map is neither injective nor surjective. It is not injective because two tuples can represent the same scalar or matrix. Two model counting instances can have the same weighted model count. It is also not surjective because not every matrix shape can be represented. Matrices that can be represented have the shape $q^m \times q^n$, so a $3 \times 2$ matrix cannot be represented, for instance. This is a limitation that arises from the Kronecker product operation on matrix representations, defined in Section~\ref{sec:repr_semantics}.

\subsubsection{Equivalence of representations}
\label{sec:repr_equiv}

Checking if two tuples represent the same value is \#\P-hard in general, since it requires solving weighted model counting instances exactly.  Despite this, it is useful to define the representation denotational semantics as a map to equivalence classes of representations, rather than the representations themselves. For this, we define the equivalence relation $\sim$ on $\Rep$ as follows:
\begin{align}
    r_1 \sim r_2 \quad\Longleftrightarrow\quad \rep(r_1) = \rep(r_2)
\end{align}
This relation induces an injective map $\rep^\#: \Rep/{\sim} \to \mathbb F \cup \Mat(\mathbb F)$. We denote the equivalence class of a representation $r$ under this relation with $[r]$.

\subsubsection{Finding equivalent representations}

We define the representation semantics in the next section as a map from type derivations to classes of representations. This is done inductively. Hence, we can have two classes of representations, and need to combine these in some way to get a new class. We do this by using representatives of the classes. In the rules in Section~\ref{sec:repr_semantics}, we introduce two types of constraints: Constraints on the domains of the representations and a constraint $\WMC(\top,W) \neq 0$.

\paragraph{Domain constraints}

We put some requirements on the domains of these representatives (e.g., the domains need to be disjoint) to combine them. Finding representations that conform to these restrictions is possible by substituting variables in the representations.

We illustrate this with an example: Suppose we define the representation semantics of $\denot{s_1 \cdot s_2}_r$ inductively. Then we already have two representatives $(\phi_1, W_1)$ and $(\phi_2, W_2)$ with $\denot{s_1}_r = [(\phi_1, W_1)]$ and $\denot{s_2}_r = [(\phi_2, W_2)]$. Now we want to combine them by using $\denot{s_1 \cdot s_2}_r = [(\phi_1 \land \phi_2, W_1 \cup W_2)]$. A requirement for this to work is that the domains of $W_1$ and $W_2$ are disjoint. We can accomplish this by substituting variables in the representation $(\phi_2, W_2)$ with fresh ones. This can be implemented efficiently.

\paragraph{Requiring \texorpdfstring{$\WMC(\top,W) \neq 0$}{WMC(T,W) not equal to 0}}

Note that $\WMC(\top, W)$ can be written as
\begin{align}
    \WMC(\top, W) = \prod_{v \in V} (W(\overline v) + W(v))
\end{align}
This quantity can only be $0$ if, for some variable $v \in V$, we have $W(\overline v) + W(v) = 0$. If we have $W(\overline v) = W(v) = 0$, then $\WMC(\phi,W) = 0$ for any Boolean formula $\phi$. Hence we have $\rep(\phi,W) = \rep(\bot,W_0)$, with $W_0: \varnothing \times \mathbb B \to \mathbb F$. Note that $\WMC(\top,W_0) = 1$.

If $W(v) + W(\overline v) \neq 0$, we instead introduce a fresh variable $v'$. We add $v \leftrightarrow v'$ to the formula $\phi$, and introduce the weight function $W': (V \cup \{v'\}) \times \mathbb B \to \mathbb F$ that is the same as $W$ on $V$, except for $W'(\overline v)=2W(\overline v)$, $W'(\overline v') = 1/2$, and $W'(v') = 1$.

Using these two methods, for every model counting instance $(\phi, W)$, we can efficiently find an equivalent instance $(\phi',W')$ with $\WMC(\top,W') \neq 0$. Note that these methods can also be applied to matrix representations. 

\subsection{Representation denotational semantics}
\label{sec:repr_semantics}

In this section, we introduce the representation denotational semantics $\denot{\cdot}_r$ for all scalar and matrix type expressions. Like with the value semantics, the representation semantics are defined on proof trees of type derivations. We refrain from proving the correctness of these operations here, but proofs can be found in Appendix~\ref{chapter:operation_proofs}.

\subsubsection{Scalar representations}

Scalar expressions are mapped to classes of equivalent scalar representations, denoted as $[(\phi, W)]$. Scalar constants form the basis of scalar expressions. These are mapped to the classes of representations of the same value $\alpha$. For the sake of implementation, we give an explicit element of this class:
\begin{operator}{Scalar constant}{}
    \begin{align}
        \denot{\alpha}_r &= [(x, W_\alpha)]
    \end{align}
    
    where $W_\alpha: \{x\} \times \mathbb B \to \mathbb F$ is a constant function $\alpha$.
\end{operator}
For model counting instances with no variables in common, the model counts can be multiplied by combining them as follows:
\begin{operator}{Scalar multiplication}{}
    \begin{align}
        \left.\begin{array}{ll}
            \denot{s_1}_r = [(\phi_1, W_1)] \\
            \denot{s_2}_r = [(\phi_2, W_2)] \\
            \dom(W_1) \cap \dom(W_2) = \varnothing
        \end{array}\right\} \quad \Longrightarrow \quad \denot{s_1 \cdot s_2}_r = [(\phi_1 \land \phi_2, W_1 \cup W_2)]
    \end{align}
    Here $W_1 \cup W_2$ indicates the union of two functions with disjoint domains $f_1: X_1 \to Y_1$ and $f_2: X_2 \to Y_2$ to a function $f_1 \cup f_2: X_1 \cup X_2 \to Y_1 \cup Y_2$.
    
    Category theoretically, this is the map $f_1 \sqcup f_2$ from the coproduct of $X_1$ and $X_2$ to $Y_1 \cup Y_2$. Defining it like this would drop the requirement for the domains to be disjoint. However, there are restrictions in other rules that would make working with this definition difficult.
\end{operator}
Adding scalars is more involved, as there is no property of weighted model counting instances that allows for easily adding results. We add a control variable $c$ that points to either $\phi_1$ or $\phi_2$ as the formula that needs to hold. The other formula does not need to hold, meaning we get a model count that is multiplied by a factor $\WMC(\top,W_i)$. We divide by this quantity by scaling the weights of $c$ appropriately. As outlined before, equivalent representations with $\WMC(\top,W) \neq 0$ can be found efficiently.

\begin{operator}{Scalar addition}{}
    \begin{gather}
        \begin{gathered}
            \left.\begin{array}{ll}
                \denot{s_1}_r = [(\phi_1, W_1)] \\
                \denot{s_2}_r = [(\phi_2, W_2)] \\
                \dom(W_1) \cap \dom(W_2) = \varnothing \\
                c \not\in \dom(W_1) \cup \dom(W_2) \\
                \WMC(\top,W_1) \neq 0, \
                \WMC(\top,W_2) \neq 0
            \end{array}\right\} \quad \Longrightarrow \\[.3cm]
            \denot{s_1 + s_2}_r = [((\overline c \to \phi_1) \land (c \to \phi_2), W_1 \cup W_2 \cup W_c)]
        \end{gathered}
    \end{gather}
    where $W_c: \{c\} \times \mathbb B \to \mathbb F$ with $W(c) = 1 / \WMC(\top,W_1)$ and $W(\overline c) = 1 / \WMC(\top,W_2)$.
\end{operator}

A field endomorphism $f$ has the property that $\WMC(\phi, f \circ W) = f(\WMC(\phi, W))$ for any weight function $W$ and formula $\phi$, which is why it is introduced in the syntax of our language.
\begin{operator}{Field endomorphism on a scalar}{}
    \begin{gather}
        \begin{gathered}
            \denot{s}_r = [(\phi, W)] \quad \Longrightarrow \quad \denot{\mathsf{apply}(f,s)}_r = [(\phi, f \circ W)]
        \end{gathered}
    \end{gather}
\end{operator}

\subsubsection{Matrix representations}

Matrix typed expressions are mapped to equivalence classes of matrix representations $[(\phi, W, x, y, q)]$, as described in Definition~\ref{def:matrix_rep}. Bras and kets form the basis of the matrix-type expressions. These are represented with formulae that fix the values of the input/output variables. The weight function is kept constant $1$. This means that there is exactly one input/output index $i$ for which the model count is $1$, and it is $0$ for all other indices.

\begin{operator}{Bra and ket}{}
    \begin{gather}
        \begin{gathered}
            \denot{\mathsf{bra}(i,q)}_r = [(x = i, W_1, x, -, q)]
        \end{gathered} \\
        \begin{gathered}
            \denot{\mathsf{ket}(i,q)}_r = [(x = i, W_1, -, x, q)]
        \end{gathered}
    \end{gather}
    where $x$ is a $q$-state variable and $W_1: \var(x) \times \mathbb B \to \mathbb F$ is constant $1$ and ``$-$'' denotes an empty string of variables.
\end{operator}

The product of two matrices is essentially the composition of two linear maps. We get the product $M_2 \cdot M_1$ by connecting the output variables of $M_1$ to the input variables of $M_2$. Figure~\ref{fig:matrix_mul} shows this schematically.

It is also necessary to add $\val_y$ for these connected variables $y$ to the formula, since $y$ no longer is an input or output of the resulting matrix $M_2 \cdot M_1$. If this were not added to the formula, it would allow for values of $y$ outside the range it can represent.

\begin{operator}{Matrix multiplication}{}
    \begin{gather}
        \begin{gathered}
            \left.\begin{array}{l}
                \denot{M_1}_r = [(\phi_1, W_1, x, y, q)] \\
                \denot{M_2}_r = [(\phi_2, W_2, y, z, q)] \\
                \dom(W_1) \cap \dom(W_2) = \var(y)
            \end{array}\right\} \quad \Longrightarrow \\[.3cm]
            \denot{M_2 \cdot M_1}_r = [(\phi_1 \land \phi_2 \land \val_y, W_1 \cdot W_2, x, z, q)]
        \end{gathered}
    \end{gather}
    The multiplication of weight functions is using the following rule for multiplying functions $f_1: X_1 \to \mathbb F$ and $f_2: X_2 \to \mathbb F$ to get a function $f_1\cdot f_2: X_1 \cup X_2 \to \mathbb F$.
    \begin{align}
        (f_1 \cdot f_2)(x) = \left\{\begin{array}{ll}
            f_1(x) & \text{if }x \not\in X_2 \\
            f_2(x) & \text{if }x \not\in X_1 \\
            f_1(x) \cdot f_2(x)\quad & \text{if }x \in X_1 \cap X_2
        \end{array}\right.
    \end{align}
\end{operator}

\begin{figure}[ht]
    \centering
    \includegraphics[page=1]{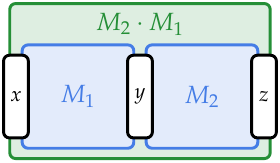}
    \caption{Diagram of multiplication operation on matrix representations.}
    \label{fig:matrix_mul}
\end{figure}

The sum of two matrices is represented in a similar way to scalars, with an extra variable that indicates which matrix should be evaluated. In this case, the input and output variables are also linked with the input and output variables of the respective matrix. Figure~\ref{fig:matrix_add} shows the operation schematically.

\begin{operator}{Matrix addition}{}
    \begin{gather}
        \begin{gathered}
            \left.\begin{array}{ll}
                \denot{M_1}_r = [(\phi_1, W_1, x_1, y_1, q)] \\
                \denot{M_2}_r = [(\phi_2, W_2, x_2, y_2, q)] \\
                \dom(W_1) \cap \dom(W_2) = \varnothing \\
                (\{c\} \cup \var(x) \cup \var(y)) \cap (\dom(W_1) \cup \dom(W_2)) = \varnothing \\
                \WMC(\top,W_1)\neq0, \ \WMC(\top,W_2)\neq0
            \end{array}\right\} \quad \Longrightarrow \\[.3cm]
            \denot{M_1 + M_2}_r = [(\phi, W_1 \cup W_2 \cup W_c \cup W_{xy}, x, y, q)]
        \end{gathered}
    \end{gather}
    with
    \begin{gather}
        \begin{aligned}
            \phi \equiv\ &(\overline c \to ((x \leftrightarrow x_1) \land (y \leftrightarrow y_1) \land \phi_1)) \\
            &\land (c \to ((x \leftrightarrow x_2) \land (y \leftrightarrow y_2) \land \phi_2))
        \end{aligned}
    \end{gather}
    and $W_c: \{c\} \times \mathbb B \to \mathbb F$ and $W_{xy}: (\var(x) \cup \var(y)) \times \mathbb B \to \mathbb F$ defined by $W_c(c) = 1/\WMC(\top,W_1)$, $W_c(\overline c) = 1/\WMC(\top,W_2)$, and $W_{xy}$ constant function $1$.
\end{operator}

\begin{figure}[ht]
    \centering
    \includegraphics[page=3]{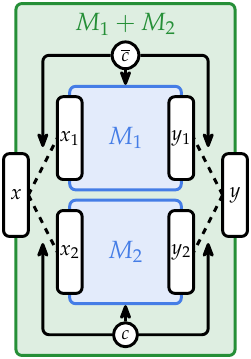}
    \caption{Diagram of addition operation on matrix representations.}
    \label{fig:matrix_add}
\end{figure}

The Kronecker product representation is constructed from the two independent representations of the matrices $M_1$ and $M_2$. The input and output variables of the two matrices are concatenated. Figure~\ref{fig:matrix_kron} shows this schematically.

\begin{operator}{Kronecker product}{}
    \begin{gather}
        \begin{gathered}
            \left.\begin{array}{l}
                \denot{M_1}_r = [(\phi_1, W_1, x_1, y_1, q)] \\
                \denot{M_2}_r = [(\phi_2, W_2, x_2, y_2, q)] \\
                \dom(W_1) \cap \dom(W_2) = \varnothing
            \end{array}\right\} \quad \Longrightarrow \\[.3cm]
            \denot{M_1 \otimes M_2}_r = [(\phi_1 \land \phi_2, W_1 \cup W_2, x_1x_2, y_1y_2, q)]
        \end{gathered}
    \end{gather}
\end{operator}

\begin{figure}[ht]
    \centering
    \includegraphics[page=2]{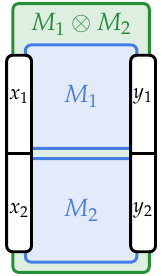}
    \caption{Diagram of Kronecker product operation on matrix representations.}
    \label{fig:matrix_kron}
\end{figure}

The multiplication of a scalar and a matrix can be represented by a conjunction of the two formulae. This operation uses the property that $\WMC(\phi \land \psi,W_1 \cup W_2) = \WMC(\phi,W_1) \cdot \WMC(\psi,W_2)$ for two formulae $\phi$ and $\psi$ for variables in the domains of $W_1$ and $W_2$ respectively (such that the domains do not overlap). Figure~\ref{fig:matrix_scamul} shows the operation schematically.

\begin{operator}{Matrix-scalar multiplication}{}
    \begin{gather}
        \begin{gathered}
            \left.\begin{array}{ll}
                \denot{s}_r = [(\phi_s, W_s)] \\
                \denot{M}_r = [(\phi, W, x, y, q)] \\
                \dom(W_s) \cap \dom(W) = \varnothing
            \end{array}\right\} \quad \Longrightarrow \quad \denot{s \cdot M}_r = [(\phi \land \phi_s, W \cup W_s, x, y, q)]
        \end{gathered}
    \end{gather}
\end{operator}

\begin{figure}[ht]
    \centering
    \includegraphics[page=4]{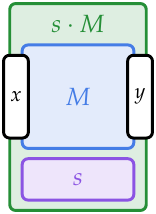}
    \caption{Diagram of multiplying a matrix $M$ with a scalar $s$, using representations for both.}
    \label{fig:matrix_scamul}
\end{figure}

The representation of the transpose of a matrix is the same, but with input and output variables swapped. The effect of this operation can be seen directly in (\ref{eq:matrix_rep}), where swapping the input and output variables replaces $(x = j) \land (y = i)$ with $(x = i) \land (y = j)$. Figure~\ref{fig:matrix_trans} shows the operation schematically.

\begin{operator}{Transpose}{}
    \begin{gather}
        \begin{gathered}
            \denot{M}_r = [(\phi, W, x, y, q)] \quad \Longrightarrow \quad \denot{\mathsf{trans}(M)}_r = [(\phi, W, y, x, q)]
        \end{gathered}
    \end{gather}
\end{operator}

\begin{figure}[ht]
    \centering
    \includegraphics[page=6]{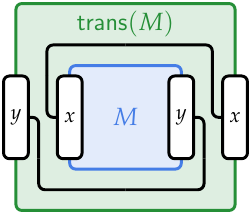}
    \caption{Diagram of the transpose of a matrix representation. Input and output variables are swapped.}
    \label{fig:matrix_trans}
\end{figure}

Applying a field endomorphism to a matrix is similar to applying it to a scalar.

\begin{operator}{Field endomorphism on a matrix}{}
    \begin{gather}
        \begin{gathered}
            \denot{M}_r = [(\phi, W, x, y, q)] \quad \Longrightarrow \quad \denot{\mathsf{apply}(f,M)}_r = [(\phi, f \circ W, x, y, q)]
        \end{gathered}
    \end{gather}
\end{operator}

The trace of a matrix can be calculated by adding a clause to the conjunction requiring the input and output variables to have the same value. In addition, we need this new input/output to be valid. Figure~\ref{fig:matrix_trace} shows the operation schematically.

\begin{operator}{Trace}{}
    \begin{gather}
        \begin{gathered}
            \denot{M}_r = [(\phi, W, x, y, q)] \quad \Longrightarrow \quad \denot{\mathsf{tr}(M)}_r = [(\phi \land (x \leftrightarrow y) \land \val_x, W)]
        \end{gathered}
    \end{gather}
\end{operator}

\begin{figure}[ht]
    \centering
    \includegraphics[page=5]{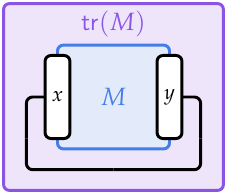}
    \caption{Diagram of taking the trace of a matrix, using a matrix representation to get a scalar representation.}
    \label{fig:matrix_trace}
\end{figure}

An entry in the matrix can be obtained by applying the definition from (\ref{eq:matrix_rep}) directly. Instead of returning the quantity $\WMC(\phi \land (x = j) \land (y = i),W)$, we return the model counting instance.

\begin{operator}{Matrix entry}{}
    \begin{gather}
        \begin{gathered}
            \denot{M}_r = [(\phi, W, x, y, q)] \quad \Longrightarrow \quad \denot{\mathsf{entry}(i,j,M)} = [(\phi \land (x = j) \land (y = i), W)]
        \end{gathered}
    \end{gather}
\end{operator}

\subsection{Correctness}
\label{sec:repr_correctness}

To use these semantics effectively, we need to be able to convert an expression to a representation, then use a model counter to get the actual matrix or scalar that is represented. We want the outcome to be the same as evaluating the expression directly (i.e., using $\denot{\cdot}_v$). What this means is that we need $\rep^\# \circ \denot{\cdot}_r = \denot{\cdot}_v$. We interpret $\denot{\cdot}_r$ and $\denot{\cdot}_v$ as maps
\begin{align}
    \denot{\cdot}_v&: \Exp \to \mathbb F \cup \Mat(\mathbb F) \\
    \denot{\cdot}_r&: \Exp \to \Rep
\end{align}
We define the set $\Exp$ of expressions that have a type.

\begin{theorem}\label{thm:repr_correctness}
    The representation semantics $\denot{\cdot}_r$ are well-defined. Furthermore, for the value semantics $\denot{\cdot}_v$ and the function $\rep^\#$ as defined in Section~\ref{sec:repr_equiv}, we have
    \begin{align}
        \rep^\# \circ \denot{\cdot}_r = \denot{\cdot}_v
    \end{align}
\end{theorem}
\begin{proof}
    See Appendix~\ref{chapter:operation_proofs}, which uses induction on the proof trees of the expression types.
\end{proof}

\subsection{Discussion}

Although most of the rules from Section~\ref{sec:repr_semantics} can be implemented efficiently, yielding a compact CNF formula, the addition rules introduce an extra variable that distributes over the already existing formulae when keeping the formulae in CNF. When doing many additions, this can cause the size of the formula to become quadratic in the number of operations.

An alternative representation $(\phi,W,x,y,q,c)$ could be introduced, which adds a ``conditional variable'' $c$. We can let the representation with $\phi \land c$ be of the original matrix, and with $\phi \land \overline c$ of the matrix with the same shape, but filled with ones. This would make the addition operation result in a more compact formula, namely
\begin{align}
    (c \to (c_1 \lor c_2)) \land (\overline c \to (\overline c_1 \land \overline c_2)) \land (\overline c_1 \lor \overline c_2) \land \phi_1 \land \phi_2
\end{align}
This requires only a constant amount of extra space per addition. However, it is not certain that the performance of the model counters would increase when using this definition, since the formula $(c \to (c_1 \lor c_2)) \land (\overline c \to (\overline c_1 \land \overline c_2))$ cannot be simplified easily.

In our method, the model counter is only called at the end of the process, once one big model counting instance is constructed. It can be beneficial to evaluate scalars and small matrices while constructing the representations. This could reduce the total size of the problems the model counter has to solve.

\section{Application: Ising Model}
\label{sec:ising}

The Ising model is a fundamental model in statistical mechanics, frequently used to study interacting systems such as ferromagnets~\cite{Brush1967}. Of particular interest is its partition function $Z_{\beta,I}$, which encodes the distribution of energy across different configurations and underlies the Boltzmann distribution.

We present two distinct methods for expressing the partition function as a weighted model counting (WMC) problem: (1) the approach of Nagy et al.~\cite{Nagy2024}, and (2) a formulation via matrix representations to which our general WMC framework in Section~\ref{sec:matrix} can be applied. We demonstrate that both approaches yield the same Boolean formula and weight function, and hence the same WMC instance.

\subsection{Definition}
\label{sec:ising_model_def}

The Ising model is defined on a finite set of sites $\Lambda$ with edge weights $J_{ij} \in \mathbb{R}$ denoting pairwise interactions, and external fields $h_i \in \mathbb{R}$. A configuration is an assignment $\sigma \colon \Lambda \to \{-1,1\}$, and its associated energy is given by the Hamiltonian:
\begin{align}
  H_I(\sigma) = -\sum_{i,j \in \Lambda} J_{ij} \sigma_i \sigma_j - \sum_{i \in \Lambda} h_i \sigma_i
\end{align}

The partition function is then defined as:
\begin{align}
  Z_{\beta,I} = \sum_{\sigma \in \{-1,1\}^{|\Lambda|}} e^{-\beta H_I(\sigma)}
\end{align}

At high temperature ($\beta \to 0$), $Z_{\beta,I}$ approximates the uniform distribution over configurations. At low temperature ($\beta \to \infty$), it concentrates on ground states minimizing $H_I(\sigma)$.

\subsection{Conversion to WMC}
\label{sec:ising_to_wmc}

Following Nagy et al.~\cite{Nagy2024}, we associate Boolean variables $x_i$ for each site and $x_{ij}$ for each interaction. The variable assignment encodes the configuration via: $x_i = 1 \Leftrightarrow \sigma_i = 1$. To enforce the interaction structure, we use the Boolean formula:
\begin{align}
  \phi = \bigwedge_{i,j \in \Lambda} (x_{ij} \leftrightarrow (x_i \leftrightarrow x_j))
\end{align}

The partition function factors as:
\begin{align}
  Z_{\beta,I} = \sum_{\sigma} \prod_{i,j} e^{\beta J_{ij} \sigma_i \sigma_j} \prod_i e^{\beta h_i \sigma_i}
\end{align}

We thus define a weight function $W$ by
\begin{align}
  \begin{split}
    W(\bar{x}_{ij}) &= e^{-\beta J_{ij}},\quad W(x_{ij}) = e^{\beta J_{ij}} \\
    W(\bar{x}_i) &= e^{-\beta h_i},\quad W(x_i) = e^{\beta h_i}
  \end{split}
\end{align}
so that $\WMC(\phi, W) = Z_{\beta,I}$.

\subsection{Matrix Representation of the Ising Model}
\label{sec:ising_matrix}

We now give a matrix formulation of the same model. Spins correspond to tensor factors in a $2^{|\Lambda|} \times 2^{|\Lambda|}$ dimensional Hilbert space. Define the diagonal Hamiltonian:
\begin{align}
  H_I = -\sum_{i,j \in \Lambda} J_{ij} Z_i Z_j - \sum_{i \in \Lambda} h_i Z_i
\end{align}
where $Z_i$ is the Pauli-$Z$ operator on qubit $i$. The partition function is obtained as:
\begin{align}
  Z_{\beta,I} = \tr(e^{-\beta H_I})
\end{align}

Since all terms in $H_I$ commute (being diagonal), we use:
\begin{align}
  e^{-\beta H_I} = \prod_{i,j} e^{\beta J_{ij} Z_i Z_j} \cdot \prod_i e^{\beta h_i Z_i}
\end{align}

Using the representation semantics from Section~\ref{sec:matrix}, we convert each matrix $e^{\theta Z}$ and $e^{\theta (Z \otimes Z)}$ into a Boolean formula with weights.

\paragraph{Encoding $e^{\theta Z}$} Define variable $x$, formula $\top$, and weight function $W(x) = e^{-\theta}$, $W(\bar{x}) = e^{\theta}$. Then:
\begin{align}
  e^{\theta Z} = \rep(\top, W, x, x, 2)
\end{align}

\paragraph{Encoding $e^{\theta (Z \otimes Z)}$} Introduce auxiliary variable $z$ and use formula $z \leftrightarrow (x \leftrightarrow y)$. Let $W(z) = e^{\theta}, W(\bar{z}) = e^{-\theta}$, and $W(x) = W(y) = 1$. Then:
\begin{align}\label{eq:tensorencoding}
  e^{\theta(Z \otimes Z)} = \rep(z \leftrightarrow (x \leftrightarrow y), W, xy, xy, 2)
\end{align}
To see that encoding~\eqref{eq:tensorencoding} is correct, note that 
$z \leftrightarrow (x \leftrightarrow y)$ is satisfied precisely 
when $z$ takes the value forced by the parity of $x$ and $y$. 
The weights on $z$ and $\bar{z}$ then contribute $e^{\theta}$ 
or $e^{-\theta}$ according to whether $x = y$ or $x \neq y$, 
reproducing the matrix elements of $e^{\theta(Z \otimes Z)}$.

\subsection{Comparison with Direct Encoding}
\label{sec:ising_reproduced}

By multiplying all representations as per Section~\ref{sec:repr_semantics}, we recover the same formula and weight function as the direct method. The variables $x_i$ and $x_{ij}$ from Nagy et al. align respectively with the encoding variables from $Z_i$ and $Z_i Z_j$ representations.

We compare the two approaches on square lattice and random graph Ising models using three WMC solvers. As shown in Figures~\ref{fig:ising_square_lattice_matrix} and~\ref{fig:ising_random_graph_matrix}, the runtimes are comparable.

Importantly, the matrix method generalizes to non-diagonal Hamiltonians such as in the quantum Ising model and Potts model, which are less amenable to direct WMC translation. Thus, our representation framework provides a reusable interface across classical and quantum models.
\begin{figure}[ht]
    \begin{minipage}[t]{0.48\textwidth}
        \centering
        \includegraphics[scale=1.35,page=3]{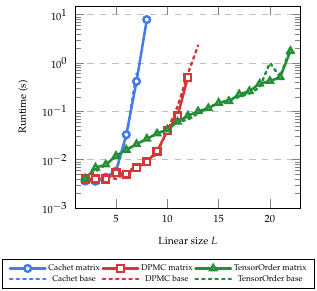}
        \caption{Runtime of calculating the partition function of an $L \times L$ square lattice Ising model with interaction strengths and external field strengths from the standard normal distribution, averaged over five runs. The problem is converted to a matrix representation from Chapter~\ref{sec:matrix}, after which the trace is calculated using a model counter. Comparison between the model counters Cachet, DPMC, and TensorOrder. Direct method from Nagy et al. in dotted lines~\cite{Nagy2024}.}
        \label{fig:ising_square_lattice_matrix}
    \end{minipage}\quad
    \begin{minipage}[t]{0.48\textwidth}
        \centering
        \includegraphics[scale=1.35,page=5]{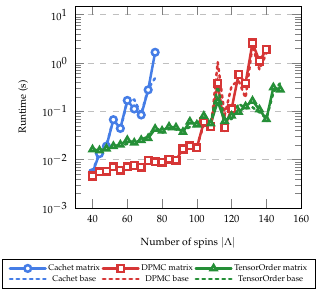}
        \caption{Runtime of calculating the partition function of a random graph Ising model for different numbers of spins (nodes), averaged over five runs. The expected degree of each node is three. The interaction strengths are uniformly chosen from $[-1, 1]$, and there is no external field. The problem is converted to a matrix representation from Chapter~\ref{sec:matrix}, after which the trace is calculated using a model counter. Comparison between the model counters Cachet, DPMC, and TensorOrder. Direct method from Nagy et al. in dotted lines~\cite{Nagy2024}.}
        \label{fig:ising_random_graph_matrix}
    \end{minipage}
\end{figure}

\section{Transverse-field Ising Model}
\label{sec:quantumising}

We next consider the quantum extension of the Ising model, where the Hamiltonian may contain non-diagonal components. In quantum mechanics, the Hamiltonian governs the time evolution of the state $\ket{\psi(t)}$ via the Schrödinger equation:
\begin{align}
    H\ket{\psi(t)} = i\hbar \frac{d}{dt}\ket{\psi(t)}.
\end{align}
The quantum partition function is defined analogously to the classical case:
\begin{align}
    Z_\beta = \mathrm{Tr}\left(e^{-\beta H}\right) = \mathrm{Tr}\left(\sum_{k=0}^\infty \frac{(-\beta H)^k}{k!}\right).
\end{align}
In the quantum case, the difficulty arises from the non-commutative nature of the Hamiltonian components. When $A$ and $B$ do not commute, we no longer have $e^{A+B} = e^Ae^B$, complicating the evaluation.

This computation plays a central role in understanding phase transitions and calculating the Helmholtz free energy:
\begin{align}
    F = -\frac{1}{\beta} \log Z_\beta.
\end{align}

\subsection{Model Definition}
\label{sec:tfim-definition}

We adopt the transverse-field Ising model introduced by Suzuki~\cite{Suzuki1976}, where interaction strengths vary pairwise but external field strengths are uniform across sites. This is a genuine quantum model and extends the classical Ising model from Section~\ref{sec:ising}, though it is not the most general quantum spin model.

\begin{definition}[Transverse-field Ising model]\label{sec:transverse_field_ising_def}
    A transverse-field (quantum) Ising model is a tuple $Q = (\Lambda, J, \mu_z, \mu_x)$ with:
    \begin{itemize}
        \item $\Lambda$: a finite set of sites;
        \item $J: \Lambda^2 \to \mathbb{R}$: symmetric coupling strengths;
        \item $\mu_z, \mu_x \in \mathbb{R}$: global field strengths.
    \end{itemize}
    The Hamiltonian is given by
    \begin{align}
        H_Q = -\sum_{i,j \in \Lambda} J_{ij} Z_i Z_j - \mu_z \sum_{i \in \Lambda} Z_i - \mu_x \sum_{i \in \Lambda} X_i,
    \end{align}
    where $Z_i$ and $X_i$ are Pauli matrices applied at site $i$.  Furthermore, we distinguish terms that contain each kind of Pauli matrices: \begin{align}
    H_{Q,Z} &= -\sum_{i,j \in \Lambda} J_{ij} Z_iZ_j - \mu_z \sum_{i \in \Lambda} Z_i \\
    H_{Q,X} &= -\mu_x \sum_{i \in \Lambda} X_i
\end{align}
    The partition function at inverse temperature $\beta > 0$ is
    \begin{align}
        Z_{\beta,Q} = \mathrm{Tr}(e^{-\beta H_Q}).
    \end{align}
\end{definition}

\paragraph{Example (Two-spin system)}
Let $\Lambda = \{1,2\}$ with $J_{12}=1$, $\mu_z=0$, $\mu_x=1$. Then the Hamiltonian matrix is
\begin{align}
    H_Q = -Z_1Z_2 - X_1 - X_2 =
    \begin{bmatrix}
        1 & 1 & 1 & 0 \\
        1 & -1 & 0 & 1 \\
        1 & 0 & -1 & 1 \\
        0 & 1 & 1 & 1
    \end{bmatrix}.
\end{align}
Evaluating the partition function at $\beta=1$ yields
\begin{align}
    Z_{\beta,Q} &\approx \mathrm{Tr}
    \begin{bmatrix}
        1.52 & -2.07 & -2.07 & 1.15 \\
        -2.07 & 4.76 & 2.04 & -2.07 \\
        -2.07 & 2.04 & 4.76 & -2.07 \\
        1.15 & -2.07 & -2.07 & 1.52
    \end{bmatrix} \approx 12.55.
\end{align}

\subsection{Trotterization and Encoding}

We use Trotterization to approximate the exponential of a non-commuting sum:
\begin{align}
    e^{A+B} = \lim_{k \to \infty} \left(e^{A/k} e^{B/k}\right)^k.
\end{align}
With $H_Q = H_{Q,Z}+H_{Q,X}$, we can approximate the partition function as
\begin{align}
    Z_{\beta,Q} \approx \tr\left(\left(e^{-\beta\frac{H_{Q,Z}}{k}} \cdot e^{-\beta\frac{H_{Q,X}}{k}}\right)^k\right),
\end{align} with increasing degree of accuracy as $k$ increases.
Another useful property of the matrix exponential is that, for an invertible matrix $P$ and any matrix $A$, we have $e^{P^{-1}AP} = P^{-1}e^AP$. In our problem, the Pauli-$X$ matrix can be diagonalized as $X = HZH$, where $H$ is the involutory Hadamard operator
\begin{align}
    H = \frac1{\sqrt2}\begin{bmatrix}
        1 & 1 \\
        1 & -1
    \end{bmatrix}
\end{align}
Define the sum of Pauli-$Z$ matrices
\begin{align}
    H_{Q,X}' &= -\mu_x\sum_{i \in \Lambda} Z_i
\end{align}
This is the same as $H_{Q,X}$, but with all Pauli-$X$ matrices replaced with a Pauli-$Z$. Therefore, we have $H_{Q,X} = H^{\otimes |\Lambda|}H_{Q,X'}H^{\otimes |\Lambda|}$, which gives
\begin{align}
    Z_{\beta,Q} &\approx \tr\left(\left(e^{-\beta\frac{H_{Q,Z}}{k}} \cdot e^{-\beta\frac{H^{\otimes|\Lambda|}H_{Q,X}'H^{\otimes|\Lambda|}}{k}}\right)^k\right) \\
    &= \tr\left(\left(e^{-\beta\frac{H_{Q,Z}}{k}}H^{\otimes|\Lambda|}e^{-\beta\frac{H_{Q,X}'}{k}}H^{\otimes|\Lambda|}\right)^k\right)
\end{align}
Note that we are left with only Hadamard matrices and exponentials of Pauli-$Z$ matrices.
Each exponential is diagonal and can be encoded as a weighted model counting instance as in Section~\ref{sec:ising}. Hadamard gates are encoded using~\cite{Mei2024}:
\begin{align}
    (r \leftrightarrow (x \land y), W, x, y, 2),
\end{align}
with $W$ assigning $-1$ to $r$ and $1$ elsewhere. Kronecker products and matrix compositions follow from Section~\ref{sec:matrix}.

\subsection{Experimental Results}

Figure~\ref{fig:quantum_ising} benchmarks the DPMC model counter on TFIM instances under Trotterization. Performance degrades as graph density increases. This is partly due to lack of decomposition: sparse graphs yield disconnected components, enabling logical formulae to decompose into conjunctions with disjoint variables.

In contrast, the \texttt{expm} routine from SciPy~\cite{SciPyExpm}, based on scaling and squaring~\cite{Al-Mohy2010}, avoids Trotterization. Figure~\ref{fig:quantum_ising_scipy} illustrates runtime scaling with the number of qubits.

\begin{figure}[ht]
    \begin{subfigure}[t]{0.48\textwidth}
        \centering
        \includegraphics[scale=1.35,page=2]{src/tikz/quantum_ising_model.pdf}
        \caption{Average degree $1$}
        \label{fig:quantum_ising_1}
    \end{subfigure}\quad
    \begin{subfigure}[t]{0.48\textwidth}
        \centering
        \includegraphics[scale=1.35,page=1]{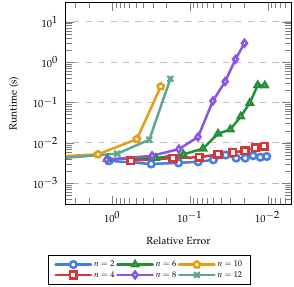}
        \caption{Average degree $3$}
        \label{fig:quantum_ising_3}
    \end{subfigure}
    \caption{Runtime vs. relative error for TFIM partition function estimation using DPMC and Trotterization. Random graphs with Gaussian couplings. Only DPMC runtime is reported.}
    \label{fig:quantum_ising}
\end{figure}

\begin{figure}[ht]
    \centering
    \includegraphics[scale=1.35,page=4]{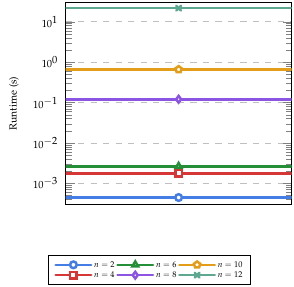}
    \caption{Runtime of SciPy \texttt{expm} method for TFIM partition function. Uses same y-axis scale as Figure~\ref{fig:quantum_ising}.}
    \label{fig:quantum_ising_scipy}
\end{figure}

\section{Potts Model}
\label{sec:potts}

We generalize the Ising model by allowing each site to take on $q \ge 2$ states. In the \emph{generalized Potts model}, interactions depend on each pair of site states, making it applicable to tasks such as image segmentation~\cite{Krahenbuhl2011,Dann2012,Portela2013} and protein modeling~\cite{Levy2016,Li2023}. We describe both the generalized and standard variants, and how each maps to weighted model counting (WMC).

\subsection{Definition}
\label{sec:potts_model_def}

Let $\Lambda$ be the set of sites and $q$ the number of possible states per site. A configuration is a map $s \colon \Lambda \to \{0,\dots,q-1\}$.

\begin{definition}
A \emph{generalized Potts model} is a tuple $P = (\Lambda, J, h, q)$ with:
\begin{itemize}
    \item $J: \Lambda^2 \times \{0,\dots,q-1\}^2 \to \mathbb{R}$: interaction strength between sites $i,j$ in states $s_i,s_j$;
    \item $h: \Lambda \times \{0,\dots,q-1\} \to \mathbb{R}$: external field on site $i$ in state $s_i$.
\end{itemize}
The Hamiltonian is
\begin{align}
    H_P(s) = -\sum_{i,j} J_{ij}(s_i, s_j) - \sum_i h_i(s_i),
\end{align}
and the partition function is
\begin{align}
    Z_{\beta,P} = \sum_{s} e^{-\beta H_P(s)}.
\end{align}
\end{definition}

\paragraph{Standard Potts Model} This special case has no external field and uses $J$ only when $s_i = s_j$ for neighbors $(i,j) \in E$ where $E$ is the set of edges:
\begin{align}
    H_P(s) = -J \sum_{(i,j) \in E} \mathbb{1}\{s_i = s_j\}.
\end{align}

\begin{example}[Three-site Standard Potts Model]
Let $q = 3$, sites $A$, $B$, $C$, edges $A$--$B$, $B$--$C$, and $J = 4$. Then
\begin{align}
    H_P(s) = -4 \cdot \mathbb{1}\{s_A = s_B\} - 4 \cdot \mathbb{1}\{s_B = s_C\},
\end{align}
with $Z_{\beta,P} = \sum_{s} e^{-H_P(s)}$. Evaluating this sum for all $3^3 = 27$ configurations yields $Z_{\beta,P} \approx 9610.05$ at $\beta = 1$.
\end{example}

\subsection{Encoding Standard Potts Model as WMC}

The Hamiltonian can be written using diagonal matrices:
\begin{align}
    H_P = -J \sum_{(i,j) \in E} M_{ij},
\end{align}
where $M = \sum_{k=0}^{q-1} \ket{k,k}\bra{k,k}$, i.e., diagonal entries equal $1$ iff $s_i = s_j$. The partition function becomes:
\begin{align}
    Z_{\beta,P} = \mathrm{tr}\left(\prod_{(i,j) \in E} e^{\beta J M_{ij}}\right).
\end{align}
Each $e^{\beta JM_{ij}}$ has diagonal entries $e^{\beta J}$ if $s_i = s_j$ and 1 otherwise. We encode this using:
\begin{align}
    (z \leftrightarrow (x \leftrightarrow y), W, x, y), \quad W(z) = e^{\beta J}.
\end{align}

\subsubsection{Empirical Comparison of Solvers}

Figure~\ref{fig:potts_solvers} benchmarks Cachet, DPMC, and TensorOrder on random graph Potts models. At $q=3$, DPMC outperforms TensorOrder. At $q=4$, TensorOrder scales better on larger instances, possibly due to encoding differences.  This also lines up with results from Nagy et al. \cite{Nagy2024} where TensorOrder works better on larger instances too.  

\subsubsection{Encoding Comparison}

Figure~\ref{fig:potts_q_varrep} compares logarithmic (uses $\lceil\log_2 q\rceil$ bits) and order encodings of the input/output vars of the matrices (uses $q-1$ bits). For small $q$, order encoding is competitive. For large $q$, logarithmic encoding is more compact and efficient.

\subsection{Encoding Generalized Potts Model as WMC}

We encode each term of $H_P$ using matrices $M(s_i,s_j)$ and $N(s_i)$ with only one nonzero diagonal entry:
\begin{align}
    H_P &= -\sum_{i,j} \sum_{s_i,s_j} M(s_i,s_j)_{ij} - \sum_i \sum_{s_i} N(s_i)_i, \\
    Z_{\beta,P} &= \mathrm{Tr}\left(\prod_{i,j} \prod_{s_i,s_j} e^{\beta M(s_i,s_j)_{ij}} \prod_i \prod_{s_i} e^{\beta N(s_i)_i}\right).
\end{align}

We encode the terms using:
\begin{align}
    (z \leftrightarrow (x = s_i \land y = s_j), W), \quad W(z) = e^{\beta J_{ij}(s_i,s_j)}, \\
    (z \leftrightarrow (x = s_i), W), \quad W(z) = e^{\beta h_i(s_i)}.
\end{align}

\paragraph{Discussion} If only a few $J_{ij}(s_i,s_j)$ are nonzero, this encoding is tractable. In general, the standard Potts encoding is more compact, requiring fewer clauses than the generalized form, and can be more efficiently compiled.

\begin{figure}[ht]
    \centering
    \begin{subfigure}[t]{0.48\textwidth}
        \includegraphics[scale=1.35,page=2]{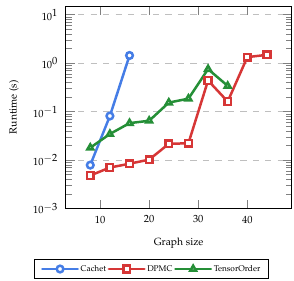}
        \caption{$q = 3$}
    \end{subfigure}
    \begin{subfigure}[t]{0.48\textwidth}
        \includegraphics[scale=1.35,page=3]{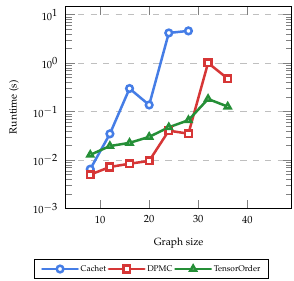}
        \caption{$q = 4$}
    \end{subfigure}
    \caption{Runtime comparison of model counters on standard Potts models. Logarithmic encoding, 5 runs averaged, edge degree $4$.}
    \label{fig:potts_solvers}
\end{figure}

\begin{figure}[ht]
    \centering
    \begin{subfigure}[t]{0.48\textwidth}
        \includegraphics[scale=1.35,page=4]{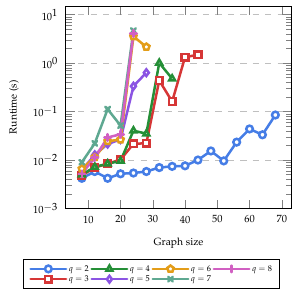}
        \caption{Logarithmic encoding}
    \end{subfigure}
    \begin{subfigure}[t]{0.48\textwidth}
        \includegraphics[scale=1.35,page=5]{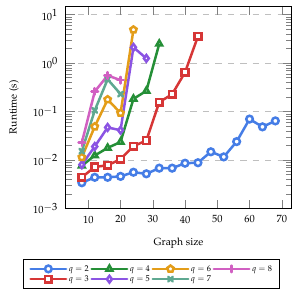}
        \caption{Order encoding}
    \end{subfigure}
    \caption{Encoding comparison on Potts partition computation using DPMC. Log encoding outperforms at higher $q$.}
    \label{fig:potts_q_varrep}
\end{figure}


\section{Related work}

\subsection{D-Hammer}

Xu et al.~\cite{Xu2025} introduced D-Hammer: A tool that can check the equivalence of quantum expressions using labeled Dirac notation. For this, they introduce rewriting rules to normalize terms. The type system and syntax they use are similar to those we introduced in Section~\ref{sec:matrix}.  This work is itself based on their earlier work on DiracDec~\cite{Xu2025b}, which uses plain Dirac notation. Their implementation of D-Hammer can be considered a generalization of $ZX$-calculus~\cite{Coecke2011}, extending it with various operations on Hilbert spaces. This comes at a performance cost on problems that can be encoded using both D-Hammer and the $ZX$-calculus.

Our work uses a subset of the operators that can be expressed using D-Hammer, since the aim of our work is to evaluate expressions written in Dirac notation, rather than checking more general equivalences. We specifically limited ourselves to common matrix operations that can be expressed relatively easily using Boolean formulae. For this reason we introduced a new (more limited) language than D-Hammer.

At the same time, our work allows for the more general notion of qudits, rather than qubits. Furthermore, our theoretical framework does not support labeling matrices. However, our implementation does have this functionality.

\subsection{Category theory}
\label{sec:cat_theory}

In general, monoidal categories~\cite{Selinger2011} can be used to provide a syntax for the sequential and parallel composition of matrices (i.e., multiplication and Kronecker product). This allows for the use of the rich field of category theory in the language definition of Chapter~\ref{sec:matrix}. Villoria et al.~\cite{villoria2024enriching} showed that, using enrichment, these operations can be extended to include other algebraic operations like convex combinations. This can be useful in simulating noise quantum circuits, for example. The language defined in Chapter~\ref{sec:matrix} can be seen as explicit syntax for enriched monoidal categories, although there are two important features of our syntax that allow the translation to WMC and make the language easy to use. First of all, we no not work with abstract finite dimensional vector spaces, but rather with the concrete choice of Cartesian spaces $\mathbb{F}^n$. On the side of monoidal categories, this amounts to working with PROPs, which are one-sorted monoidal categories~\cite{Lack04:ComposingPROPs,BCR18:PropsNetworkTheorya} and provide the foundation for a syntax where explicit coherence morphisms for associativity etc. of the tensor are not needed. Second, we add to the syntax matrix-specific operations, like applying field morphisms, and algebraic operations, like addition, to provide a usable syntax and an efficient translation to WMC instances.

\subsection{Quantum circuit simulation using WMC}

Mei et al.~\cite{Mei2024b, Mei2024c} showed that model counting can be used in quantum computing for simulating circuits and equivalence checking. They showed that quantum states can be encoded using variables~\cite{Mei2024}. Gates can then be encoded by expressing a relationship between input and output variables, which are the states before and after the gate is applied. More recently, Zak et al.~\cite{Zak2025} extended this work by showing that model counting techniques can be used for the synthesis of quantum circuits. We also build on this work by generalizing the expression of quantum operators using weighted model counting, and consequently applying it to practical applications.

\subsection{Ising model partition function}

Nagy et al.~\cite{Nagy2024} showed how the Ising model partition function problem can be converted to a WMC instance. They proved that existing model counters like TensorOrder~\cite{Dudek2021} show competitive performance compared to existing techniques. In addition, they relate the problem of calculating the partition function to $\#\mathrm{CSP}$, using powerful theoretical tools to gain insights into where the hardness of the problem comes from. We reproduced some of the experimental results and performed experiments on the transverse-field Ising and Potts models. We did this while using the matrix representations instead of converting these problems directly to WMC.

\subsection{Hamiltonian simulation using decision diagrams}

Sander et al.~\cite{Sander2023} showed how Hamiltonian simulation can be performed using decision diagrams. At every node in the decision diagram, four outgoing edges represent the four different quadrants of a square $2^n \times 2^n$ matrix. This is then applied recursively to these quadrants. The edges contain weights, such that the value at a specific entry in the matrix can be found by traversing the decision diagram from the root to a leaf. Their technique of splitting up the matrix into quadrants is similar to our technique of using strings of input and output variables. We use a weighted model counting representation, where Sander et al. use decision diagrams.

\subsection{Model counters}

Model counters can employ several different techniques to calculate weighted model counts efficiently. The three solvers we used in this work are: (1) Cachet~\cite{Sang2005}, an older tool that uses clause learning, (2) DPMC~\cite{Dudek2020}, which employs a dynamic programming technique, and (3) TensorOrder~\cite{Dudek2021}, which uses tensor-network contraction to solve WMC problems. Other successful model counters include Ganak~\cite{Sharma2019} and ProCount~\cite{Dudek2021ProCount}.
Ganak is a natural candidate for future experiments; 
we leave a systematic comparison to future work while ProCount targets weighted \emph{projected} model counting, 
which is a different task from the plain WMC used in our framework.

\subsection{Comparison to Tensor Networks}
Tensor-network approaches such as matrix product states (MPS) excel when the underlying quantum state has low entanglement and admits efficient contraction~\cite{Schollwoeck2011,Orus2014}. Our method, in contrast, leverages symbolic factorization and constraint sharing, making it advantageous in settings where the structure of the problem dominates over entanglement properties. In such cases, weighted model counting can reuse large portions of the computation across different instances or parameter regimes, producing exact or certified results at little additional cost. For example, once a diagram is compiled, scanning over parameters is essentially free, while tensor-network methods must redo the contraction. Thus, rather than competing directly with MPS on large generic systems, our approach offers a different perspective of efficiency: structural reuse and explainability.

\section{Conclusion}
\label{sec:conclusion}

We presented a general framework for encoding quantum and classical operators as Boolean formulae with weight functions, enabling a reduction from problems written in Dirac notation to weighted model counting (WMC) instances. This framework facilitates the classical evaluation of otherwise intractable quantum problems by leveraging the power of modern model counters. We demonstrated its effectiveness on two nontrivial physical models: the transverse-field Ising model (quantum) and the Potts model (classical).

Our approach generalizes prior work by supporting arbitrary $q^n \times q^m$ matrices, beyond specific domains such as quantum circuit evaluation. We formally defined a representation language and semantics for linear-algebraic operations,such as matrix multiplication, addition and trace, and proved the correctness of these constructions. The framework provides a reusable foundation that separates the problem modeling from the WMC encoding, enabling future applications across various domains.

We validated the framework by recovering known results for the classical Ising model~\cite{Nagy2024}, and then extended its reach to new problems: we encoded and evaluated the Potts model partition function and approximated the partition function of the transverse-field Ising model using Trotterization. These results show that WMC can be applied beyond qubit systems and quantum circuit simulation.

\subsection{Evaluation}

First, we construct a generic framework that encodes arbitrary matrices with defined semantics for basic operations. The representation is currently restricted to a specific matrix shape and exhibits quadratic size blow-up under addition, but already marks a significant generalization compared to approaches that target specific model classes, such as the reduction of the classical Ising partition function to WMC \cite{Nagy2024}, or WMC-based inference in graphical models \cite{Sang2005, Chavira2008}.

Second, we implement the framework in \texttt{DiracWMC}~\cite{DiracWMC} and successfully applied it to compute the partition functions of the transverse-field Ising and Potts models. The performance on the Potts model is comparable to that of the Ising model, while the quantum model was tractable only for small systems. Nevertheless, we expect scalability to improve with advances in WMC solvers.

\subsection{Future Work}

Several directions remain for future exploration:
\begin{itemize}
    \item \textbf{Tensor extensions:} Extending the framework to represent higher-order tensors would broaden its applicability in quantum many-body physics.
    \item \textbf{Alternative representations:} New encodings could address limitations such as formula size growth under matrix addition, potentially yielding performance benefits. For example, one could add a new operator that would specifically encode addition into the Boolean formula. This would, however, mean model counters have to be adapted to take this new operator into account, or an extra translation step is needed. 
    \item \textbf{Categorical formulations:} Reformulating the framework using monoidal categories may offer a richer mathematical foundation and simplify correctness proofs.
    \item \textbf{Complexity and compilation:} Understanding the complexity of intermediate representations and optimizing for model counter performance could improve scalability.
    \item \textbf{Max-WMC and ground states:} Generalizations to maximum weighted model counting could enable applications to optimization problems such as ground-state estimation.
\end{itemize}


\section*{Notation}
\addcontentsline{toc}{section}{Notation}  

Below is a table with the notation used in this work, along with the section where the notation is introduced.

\renewcommand*{\arraystretch}{1.4}
\begin{longtable}{|p{3cm}|p{1.4cm}|p{9.5cm}|}
    \hline
    \textbf{Notation} & \textbf{Intr.} & \textbf{Meaning} \\
    \hline\hline \endhead

    \hline
    \endfoot

    $\mathbb B$ & \ref{sec:boolean_logic} & The set of binary values $\{0, 1\}$. \\
    $\phi[\tau]$ & \ref{sec:boolean_logic} & A Boolean formula $\phi$ over a set of variables $V$ evaluated for an assignment $\tau: V \to \mathbb B$. \\
    $\mathbb 1\{c\}$ & \ref{sec:boolean_logic} & Indicator function returning $1$ if the condition $c$ is true, and $0$ otherwise. \\
    $\overline v$ & \ref{sec:boolean_logic} & The negation of a Boolean variable $v$. \\
    $\phi \equiv \psi$ & \ref{sec:boolean_logic} & Logical equivalence of Boolean formulae $\phi$ and $\psi$. \\
    $\mathbb F$ & \ref{def:wmc} & Some arbitrary field, which is assumed to be the same field throughout this work. \\
    $\WMC(\phi, W)$ & \ref{def:wmc} & Weighted model count of $\phi$ with respect to $W$. \\
    $\top$ & \ref{def:wmc} & ``Always true'' Boolean formula. Often referred to as ``top''. \\
    $\bot$ & \ref{def:wmc} & ``Always false'' Boolean formula. Often referred to as ``bottom''. \\

    $\mathcal S$ & \ref{sec:syntax} & Scalar type in the language from Section~\ref{sec:matrix}. \\
    $\mathcal M(q,m\to n)$ & \ref{sec:syntax} & Type of a $q^n \times q^m$ matrix in the language from Section~\ref{sec:matrix}. \\
    $\mathsf{tr}(M)$ & \ref{sec:syntax} & Matrix trace expression in the language from Section~\ref{sec:matrix}. \\
    $\mathsf{entry}(i,j,M)$ & \ref{sec:syntax} & Expression for the matrix entry at row $i$ and column $j$, in the language from Section~\ref{sec:matrix}. \\
    $\mathsf{apply}(f,s)$ & \ref{sec:syntax} & Expression for the application of a field endomorphism on a scalar in the language from Section~\ref{sec:matrix}. \\
    $\mathsf{bra}(q,i)$ & \ref{sec:syntax} & Expression in the language from Section~\ref{sec:matrix} for the length-$q$ row matrix $\bra{i}_q$. \\
    $\mathsf{ket}(q,i)$ & \ref{sec:syntax} & Expression in the language from Section~\ref{sec:matrix} for the length-$q$ column matrix $\ket{i}_q$. \\
    $\mathsf{trans}(M)$ & \ref{sec:syntax} & Expression for the transpose of a matrix in the language from Section~\ref{sec:matrix}. \\
    $\mathsf{apply}(f,M)$ & \ref{sec:syntax} & Expression for the entry-wise application of a field endomorphism on a matrix. \\
    $\vdash e : T$ & \ref{sec:type_system} & Expression $e$ has type $T$. \\
    $\denot{e}_v$ & \ref{sec:value_semantics} & Value denotational semantics of an expression $e$. \\
    $\tr(M)$ & \ref{sec:value_semantics} & The trace of a matrix. \\
    $\bra{i}_q$, $\bra{i}$ & \ref{sec:value_semantics} & A row vector of width $q$ with a $1$ at the $i$-th position counting from $0$, and $0$ everywhere else. The $q$ is left out if it is clear from context. \\
    $\ket{i}_q$, $\ket{i}$ & \ref{sec:value_semantics} & A column vector of height $q$ with a $1$ at the $i$-th position counting from $0$, and $0$ everywhere else. The $q$ is left out if it is clear from context. \\
    $\rep(\phi,W)$ & \ref{sec:scalar_repr} & Value that the tuple $(\phi,W)$ represents, which is equal to $\WMC(\phi,W)$. \\
    $\var(v)$ & \ref{sec:matrix_repr}, \ref{chapter:encodings} & Set of Boolean variables used in the variable representation $v$. \\
    $v = n$ & \ref{sec:matrix_repr}, \ref{chapter:encodings} & Formula for equality of a variable encoding $v$ to a value $n$ \\
    $\val_v$ & \ref{sec:matrix_repr}, \ref{chapter:encodings} & Validity formula of a variable encoding $v$. \\
    $v \leftrightarrow w$ & \ref{sec:matrix_repr}, \ref{chapter:encodings} & Equality formula of two variable encodings $v$ and $w$. \\
    $\rep(\phi,W,x,y,q)$ & \ref{sec:matrix_repr} & Matrix that the tuple $(\phi,W,x,y,q)$ represents. \\
    $\Rep$ & \ref{sec:repr_map} & Set of scalar and matrix representations. \\
    $\Mat(\mathbb F)$,\newline$\Mat(\mathbb F, n \times m)$ & \ref{sec:repr_map} & Set of matrices over a field $\mathbb F$, optionally with the given shape. \\
    $[r]$ & \ref{sec:repr_equiv} & Equivalence class of a representation $r$ under the relation of equal outputs when applying the function $\rep$. \\
    $\denot{e}_r$ & \ref{sec:repr_semantics} & Representation denotational semantics of an expression $e$. \\
    $f_1 \cup f_2$ & \ref{sec:repr_semantics} & Union of two functions with disjoint domains. \\
    $f_1 \cdot f_2$ & \ref{sec:repr_semantics} & Multiplication of functions on where their domains overlap, and the value of one of the functions elsewhere. \\
    $\Exp$ & \ref{sec:repr_correctness} & Set of expressions that have a type. \\

    $\Lambda$ & \ref{sec:ising_model_def}, \ref{sec:transverse_field_ising_def}, \ref{sec:potts_model_def} & Set of sites in an Ising model, transverse-field Ising model, or Potts model. \\
    $J_{ij}$ & \ref{sec:ising_model_def}, \ref{sec:transverse_field_ising_def} & Interaction strength between sites in an Ising model, transverse-field Ising model. \\
    $h_i$ & \ref{sec:ising_model_def} & External field strnegth at site $i$ in an Ising model. \\
    $\sigma$ & \ref{sec:ising_model_def} & Configuration of spins of an Ising model. \\
    $H_I(\sigma)$ & \ref{sec:ising_model_def} & Hamiltonian of an Ising model with spin configuration $\sigma$. \\
    $Z_{\beta,I}$ & \ref{sec:ising_model_def} & Partition function of an Ising model at inverse temperature $\beta$. \\
    $X$, $Y$, $Z$ & \ref{sec:ising_matrix} & Pauli $X$, $Y$, and $Z$ matrices. \\
    $e^M$ & \ref{sec:ising_matrix} & Matrix exponential $\sum_{k=0}^\infty M^k/k!$. \\

    $\mu_x$, $\mu_z$ & \ref{sec:transverse_field_ising_def} & Transverse-field Ising model external field strengths in $X$ and $Z$ directions. \\
    $H_Q$ & \ref{sec:transverse_field_ising_def} & Transverse-field Ising model Hamiltonian matrix. \\
    $Z_{\beta,Q}$ & \ref{sec:transverse_field_ising_def} & Partition function of a transverse-field Ising model at inverse temperature $\beta$. \\

    $J_{ij}(s_i,s_j)$ & \ref{sec:potts_model_def} & Interaction strength between sites $i$ and $j$ in a generalized Potts model, given their states $s_i$ and $s_j$. \\
    $h_i(s_i)$ & \ref{sec:potts_model_def} & External field strength at site $i$ in a generalized Potts model, given its state $s_i$. \\
    $s$ & \ref{sec:potts_model_def} & Configuration of spins of a Potts model. \\
    $H_P(s)$ & \ref{sec:potts_model_def} & Hamiltonian of a Potts model with spin configuration $s$. \\
    $Z_{\beta,P}$ & \ref{sec:potts_model_def} & Partition function of a Potts model $P$ at inverse temperature $\beta$. \\
\end{longtable}

\appendix
\section{Variable encodings}
\label{chapter:encodings}

In this section, we describe several ways of encoding a $q$-state variable (A variable that can take values $\{0, \dots, q-1\}$) using Boolean variables. There are several ways of doing this, each of which is useful in certain situations~\cite{Nguyen2021, Prestwich2021, Klieber2007, Abio2014, Tanjo2011}. We will introduce several of these variable encodings that can be used in the matrix representations described in Section~\ref{sec:matrix}. For the matrix representations, we need support for several operations on these variable encodings. Hence, we introduce the following formal definition:
\begin{definition}{Variable encoding}{}
    A \textbf{variable encoding} is a tuple $v = (q, V, =)$, with $q \in \mathbb Z_{\geq 2}$ the base of the encoding, $V$ a set of variables the encoding uses, and $=: \{0, \dots, q-1\} \to \Formula(V)$ a function sending a value $n$ to a formula over $V$ that indicates the value of $v$ is equal to $n$. We use the notation $v = n$. The following should be true:
    \begin{enumerate}
        \item $(v = n) \land (v = m) \equiv \bot$ for all $n \neq m$;
        \item For every $n$, there exists exactly one $\tau: V \to \mathbb B$ such that $(v = n)[\tau]$ holds.
    \end{enumerate}
\end{definition}
In addition to this definition, we also introduce notation for formulae that indicate an encoding ``has a valid value'' and that two encodings ``have the same value''. These formulae can often be simplified, as will be done in the sections discussing different encodings.
\begin{notation}
    For encodings $v=(q, V, =)$ and $w=(q, V', =)$, use the following notation:
    \begin{itemize}
        \item $q(v) = q$;
        \item $\var(v) = V$;
        \item $\val_v \equiv \bigvee_{n=0}^{q-1} (v=n)$;
        \item $v \leftrightarrow w \equiv \bigvee_{n=0}^{q-1} (v=n \land w=n)$.
    \end{itemize}
\end{notation}

Note that a string/tuple of encodings $x=x_{k-1}\dots x_0$ can represent numbers from $\{0, \dots, q^k-1\}$ by using base-$q$ expansions. We can then use the same notation as defined above for these strings, with:
\begin{align}
    x = n &\equiv \bigwedge_{i=0}^{k-1}(x_i = n_i) \qquad\text{ for }n = \sum_{i=0}^{k-1}q^in_i\text{ with }0 \leq n_i < q \\
    \var(x) &= \bigcup_{i=0}^{k-1} \var(x_i) \\
    \val_x &\equiv \bigwedge_{i=0}^{k-1} \val_{x_i} \\
    x \leftrightarrow y &\equiv \bigwedge_{i=0}^{k-1} (x_i \leftrightarrow y_i)
\end{align}
Below we list some example encodings. Note that in our implementation, some auxiliary variables may be introduced to make the Boolean formulae more compact in CNF form. However, the structure of the encodings is largely the same.

\subsection{Logarithmic encoding}

First, we introduce an encoding that uses a logarithmic number of variables relative to the base $q$~\cite{Prestwich2021}. To be precise, we use variables $v_0, \dots, v_{k-1}$, where $k = \lceil \log_2 q\rceil$. Naturally, the set of used variables of an encoding $v$ is
\begin{align}
    \var(v) = \{v_0, \dots, v_{k-1}\}.
\end{align}
The equality formula for a number $n$ is a cube (conjunction of literals) where $v_i$ or $\overline v_i$ is present depending on whether $n_i$ from the binary representation $n=\sum_{i=0}^{k=1} 2^in_i$ is equal to $1$ or $0$ respectively.

The formula $\val_v$ can be rewritten as follows, making use of the binary representation $q - 1 = \sum_{i=0}^{k-1}2^iq_i$:
\begin{align}
    \val_v \equiv \bigwedge_{\stackrel{i \in \{0, \dots, k-1\}}{q_i=0}} \left(\overline v_i \lor \bigvee_{\stackrel{j \in \{i+1, \dots, k-1\}}{q_j=1}} \overline v_j\right).
\end{align}
For some other base-$q$ encoding $w$ with variables $w_0, \dots, w_{k-1}$, the formula $v \leftrightarrow w$ can be rewritten by comparing all variables separately:
\begin{align}
    v \leftrightarrow w \equiv \bigwedge_{i=0}^{k-1} (v_i \leftrightarrow w_i).
\end{align}

\subsection{Order encoding}
\label{sec:order_encoding}

An alternative encoding, which is beneficial in some cases for SAT solvers and model counters, is an order encoding~\cite{Abio2014}. This encoding uses $q-1$ variables $v_0, \dots, v_{q-2}$, where each variable $v_i$ should be true if the represented number is strictly larger than $i$. A big advantage of this representation is the compact formula for $\val_v$, which is a conjunction of implications, making sure no false value comes before a true one:
\begin{align}
    \val_v \equiv \bigwedge_{i=1}^{k-1} (v_i \rightarrow v_{i-1}).
\end{align}
For another base-$q$ encoding $w$ using variables $w_0, \dots, w_{k-1}$, we can again check for equality by checking all variables have the same value:
\begin{align}
    v \leftrightarrow w \equiv \bigwedge_{i=0}^{k-1} (v_i \leftrightarrow w_i).
\end{align}

\subsection{One-hot encoding}

Finally, we introduce a one-hot encoding, also known as a direct encoding~\cite{Prestwich2021}. The encoding requires exactly one out of $q$ variables to be true. The set of variables is $\var(v) = \{v_0, \dots, v_{q-1}\}$. The validity formula makes sure exactly one variable is true, which can be done with
\begin{align}
    \val_v \equiv \left(\bigvee_{i=0}^{q-1} v_i\right) \land \bigwedge_{i \neq j} (\overline v_i \lor \overline v_j).
\end{align}
Again, equality of two encodings can be checked by checking if all variables have the same value:
\begin{align}
    v \leftrightarrow w \equiv \bigwedge_{i=0}^{k-1} (v_i \leftrightarrow w_i).
\end{align}
Note that there are more efficient methods similar to this encoding~\cite{Nguyen2021,Klieber2007}. From these methods, only the addition of auxiliary variables has been (partially) used in our implementation of the matrix representations.


\section{Correctness of representation denotational semantics}
\label{chapter:operation_proofs}

In this section we prove Theorem~\ref{thm:repr_correctness}, which states that the semantics $\denot{\cdot}_r$ are well-defined and that $\rep^\# \circ \denot{\cdot}_r = \denot{\cdot}_v$. We do this using induction on the expression type proof tree. Each separate rule has its own Lemma below. Since $\denot{\cdot}_v$ is defined as evaluating the expression using the usual rules, we refrain from mentioning $\denot{\cdot}_v$ explicitly in the remainder of this section.

Section~\ref{sec:wmc_properties} lists some general properties of weighted model counting. The proofs in Section~\ref{sec:correctness_proof} rely on these properties. The proof of Theorem~\ref{sec:repr_correctness} is split up into Lemmas, one for every rule in Section~\ref{sec:repr_semantics}.

\subsection{Properties of WMC}
\label{sec:wmc_properties}

\begin{lemma}\label{lem:wmc_product_indep}
    Let $\phi_1$ and $\phi_2$ be Boolean formulae over sets of variables $V_1$ and $V_2$, respectively. Let $W_1: V_1 \times \mathbb B \to \mathbb F$ and $W_2: V_2 \times \mathbb B \to \mathbb F$ be weight functions. If $V_1 \cap V_2 = \varnothing$ we have
    \begin{align}
        \WMC(\phi_1 \land \phi_2, W_1 \cup W_2) = \WMC(\phi_1, W_1) \cdot \WMC(\phi_2, W_2)
    \end{align}
\end{lemma}
\begin{proof}
    Write $S = \WMC(\phi_1 \land \phi_2, W_1 \cup W_2)$. We have
    \begin{align}
        S &= \sum_{\tau: V_1 \cup V_2 \to \mathbb B} (\phi_1 \land \phi_2)[\tau] \prod_{v \in V_1 \cup V_2} (W_1 \cup W_2)(v,\tau(v))
    \end{align}
    The function $\tau$ can be split up into two functions $\tau_1$ and $\tau_2$ over the two domains $V_1$ and $V_2$, respectively. Since $\phi_1$ only contains variables from $V_1$, we have $\phi_1[\tau] = \phi_1[\tau_1]$, and likewise for $\phi_2$.
    \begin{align}
        S &= \sum_{\tau_1: V_1 \to \mathbb B} \sum_{\tau_2: V_2 \to \mathbb B} \phi_1[\tau_1] \phi_2[\tau_2] \left(\prod_{v \in V_1} W_1(v,\tau_1(v))\right) \left(\prod_{v \in V_2} W_2(v,\tau_2(v))\right) \\
        &= \left(\sum_{\tau_1: V_1 \to \mathbb B}\phi_1[\tau_1]\prod_{v \in V_1} W_1(v,\tau_1(v))\right) \left(\sum_{\tau_2: V_2 \to \mathbb B}\phi_2[\tau_2]\prod_{v \in V_2} W_2(v,\tau_2(v))\right) \\
        &= \WMC(\phi_1, W_1) \cdot \WMC(\phi_2, W_2)
    \end{align}
    This proves the lemma.
\end{proof}

\begin{lemma}\label{lem:wmc_disjoint_formulae}
    Let $\phi_1$ and $\phi_2$ be Boolean formulae over a set of variables $V$, with $\phi_1 \land \phi_2 \equiv \bot$ (i.e., $\phi_1 \land \phi_2$ is unsatisfiable). Let $W: V \times \mathbb B \to \mathbb F$ be a weight function. Then
    \begin{align}
        \WMC(\phi_1 \lor \phi_2, W) = \WMC(\phi_1, W) + \WMC(\phi_2, W)
    \end{align}
\end{lemma}
\begin{proof}
    Writing out the definition, we have
    \begin{align}
        \WMC(\phi_1 \lor \phi_2,W)
        &= \sum_{\tau: V \to \mathbb B} (\phi_1 \lor \phi_2)[\tau] \prod_{v \in V} W(v, \tau(v))
    \end{align}
    Because $\phi_1$ and $\phi_2$ cannot be satisfied at the same time, we can write $(\phi_1 \lor \phi_2)[\tau] = \phi_1[\tau] + \phi_2[\tau]$. Taking this outside of the entire sum gives the desired result.
\end{proof}

\begin{lemma}\label{lem:wmc_fix_var}
    Let $\phi$ be a Boolean formula, $W: V \times \mathbb B \to \mathbb F$ a weight function, and $v \in V$ a variable. Then
    \begin{align}
        \WMC(\phi,W) = \WMC(\phi \land \overline v, W) + \WMC(\phi \land v, W)
    \end{align}
\end{lemma}
\begin{proof}
    This follows from Lemma~\ref{lem:wmc_disjoint_formulae} by using $\phi \equiv (\phi \land \overline v) \lor (\phi \land v)$
\end{proof}

\begin{lemma}\label{lem:wmc_field_endomorphism}
    Let $\phi$ be a Boolean formula, $W: V \times \mathbb B \to \mathbb F$ a weight function, and $f: \mathbb F \to \mathbb F$ a field endomorphism. Then
    \begin{align}
        \WMC(\phi, f \circ W) = f(\WMC(\phi, W))
    \end{align}
\end{lemma}
\begin{proof}
    Since $f$ is a field endomorphism, it has properties $f(x+y) = f(x)+f(y)$ and $f(xy)=f(x)f(y)$. This means
    \begin{align}
        \WMC(\phi, f \circ W)
        &= \sum_{\tau: V \to \mathbb B} \phi[\tau] \prod_{v\in V} (f\circ W)(v, \tau(v)) \\
        &= \sum_{\tau: V \to \mathbb B} \phi[\tau] \cdot f\left(\prod_{v \in V} W(v,\tau(v))\right).
    \end{align}
    Note that $\phi[\tau] \in \{0,1\}$. If $\phi[\tau] = 0$, for any $x$ we have $f(\phi[\tau]\cdot x) = 0 = f(0) = \phi[\tau]\cdot f(x)$. If $\phi[\tau] = 1$ we have the same property: $f(\phi[\tau]\cdot x) = f(x) = \phi[\tau]\cdot f(x)$. Therefore, we can rewrite the equation above as
    \begin{align}
        \WMC(\phi, f \circ W)
        &= \sum_{\tau: V \to \mathbb B} f\left(\phi[\tau] \cdot \prod_{v \in V} W(v,\tau(v))\right) \\
        &= f\left(\sum_{\tau: V \to \mathbb B} \phi[\tau] \cdot \prod_{v \in V} W(v,\tau(v))\right) \\
        &= f(\WMC(\phi, W))
    \end{align}
\end{proof}

\subsection{Correctness proof}
\label{sec:correctness_proof}

Combining the Lemmas below with induction on type derivation trees proves Theorem~\ref{thm:repr_correctness}.

\begin{lemma}[Scalar constant]{}
    Let $\alpha \in \mathbb F$. Then $\denot{\alpha}_r$ is well-defined and $\rep^\#(\denot{\alpha}_r) = \alpha$.
\end{lemma}
\begin{proof}
    The fact that $\denot{\alpha}_r$ is well-defined follows directly from the definition. For $W: \{x\} \times \mathbb B \to \mathbb F$ constant $\alpha$ and $\phi \equiv x$, we have
    \begin{align}
        \rep^\#(\denot{\alpha}_r)
        &= \rep^\#[(\phi, W)] \\
        &= \rep(\phi, W) \\
        &= \WMC(\phi,W) \\
        &= \sat(\phi \land \overline x) \cdot W(\overline x) + \sat(\phi \land x) \cdot W(x) \\
        &= W(x) \\
        &= \alpha
    \end{align}
\end{proof}

\begin{lemma}[Scalar multiplication]{}
    Let $s_1$ and $s_2$ be expressions of type $\mathcal S$. Suppose $\denot{s_1}_r$ and $\denot{s_2}_r$ are well-defined. Then $\denot{s_1 \cdot s_2}_r$ is well-defined, and
    \begin{align}
        \rep^\#(\denot{s_1 \cdot s_2}_r) = \rep^\#(\denot{s_1}_r) \cdot \rep^\#(\denot{s_2}_r)
    \end{align}
\end{lemma}
\begin{proof}
    Suppose $\denot{s_1}_r = [(\phi_1, W_1)]$ and $\denot{s_2}_r = [(\phi_2, W_2)]$ with the domains of $W_1$ and $W_2$ disjoint. Then by definition
    \begin{align}
        \denot{s_1 \cdot s_2}_r = [(\phi_1 \land \phi_2, W_1 \cup W_2)]
    \end{align}
    We show that this expression is well-defined by showing the value of $\rep^\#(\denot{s_1 \cdot s_2}_r)$ is independent of the choice of $\phi_1$, $\phi_2$, $W_1$, and $W_2$. We have
    \begin{align}
        \rep^\#(\denot{s_1 \cdot s_2}_r)
        = \WMC(\phi_1 \land \phi_2,W_1 \cup W_2)
    \end{align}
    Since $W_1$ and $W_2$ have non-overlapping domains and $\phi_1$ and $\phi_2$ have variables in the domains of $W_1$ and $W_2$ respectively, Lemma~\ref{lem:wmc_product_indep} gives
    \begin{align}
        \rep^\#(\denot{s_1 \cdot s_2}_r)
        &= \WMC(\phi_1,W_1) \cdot \WMC(\phi_2,W_2) \\
        &= \rep^\#([(\phi_1, W_1)]) \cdot \rep^\#([(\phi_2, W_2)]) \\
        &= \rep^\#(\denot{s_1}_r)\cdot\rep^\#(\denot{s_2}_r)
    \end{align}
    By assumption, the above expression is well-defined.
\end{proof}

\begin{lemma}[Scalar addition]{}
    Let $s_1$ and $s_2$ be expressions of type $\mathcal S$. Suppose $\denot{s_1}_r$ and $\denot{s_2}_r$ are well-defined. Then $\denot{s_1 + s_2}_r$ is well-defined, and
    \begin{align}
        \rep^\#(\denot{s_1 + s_2}_r) = \rep^\#(\denot{s_1}_r) + \rep^\#(\denot{s_2}_r)
    \end{align}
\end{lemma}
\begin{proof}
    Suppose $\denot{s_1}_r = [(\phi_1, W_1)]$ and $\denot{s_2}_r = [(\phi_2, W_2)]$ with the domains of $W_1$ and $W_2$ disjoint, $\WMC(\top,W_1) \neq 0$, and $\WMC(\top,W_2) \neq 0$. By definition
    \begin{align}
        \denot{s_1 + s_2}_r = [((\overline c \implies \phi_1) \land (c \implies \phi_2), W_1 \cup W_2 \cup W_c)]
    \end{align}
    with $W_c: \{c\} \times \mathbb B \to \mathbb F$ defined by $W_c(c) = 1/\WMC(\top,W_1)$ and $W_c(\overline c) = 1/\WMC(\top,W_2)$. Again, we show that this is well-defined by showing $\rep^\#(\denot{s_1+s_2}_r)$ yields the same value independent of the formulae and weight functions chosen at the start.
    
    Write $W = W_1 \cup W_2 \cup W_c$ and $\psi \equiv (\overline c \implies \phi_1) \land (c \implies \phi_2)$. Then
    \begin{align}
        \rep^\#(\denot{s_1 + s_2}_r) = \WMC(\psi,W)
    \end{align}
    By Lemma~\ref{lem:wmc_fix_var} we have
    \begin{align}
        \rep^\#(\denot{s_1 + s_2}_r) = \WMC(\psi \land \overline c,W) + \WMC(\psi \land c,W)
    \end{align}
    We will show that $\WMC(\psi \land \overline c,W) = \rep^\#(\denot{s_1}_r)$, the case $\WMC(\psi \land c,W)$ is symmetric. Note that $\psi \land \overline c \equiv \phi_1 \land \overline c$. Since $\phi_1$ contains only variables variables in the domain of $W_1$, Lemma~\ref{lem:wmc_product_indep} gives
    \begin{align}
        \WMC(\psi \land \overline c,W)
        &= \WMC(\phi_1 \land \overline c,W) \\
        &= \WMC(\phi_1,W_1) \cdot \WMC(\top,W_2) \cdot \WMC(\overline c, W_c) \\
        &= \WMC(\phi_1,W_1) \cdot \WMC(\top,W_2) \cdot \frac1{\WMC(\top,W_2)} \\
        &= \WMC(\phi_1,W_1)
    \end{align}
    Similarly it can be shown that $\WMC(\psi \land c,W) = \WMC(\phi_2,W_2)$, which means
    \begin{align*}
        \rep^\#(\denot{s_1 + s_2}_r)
        &= \WMC(\phi_1,W_1) + \WMC(\phi_2,W_2) \\
        &= \rep^\#[(\phi_1, W_1)] + \rep^\#[(\phi_2, W_2)] \\
        &= \rep^\#(\denot{s_1}_r) + \rep^\#(\denot{s_2}_r)
    \end{align*}
\end{proof}

\begin{lemma}[Field endomorphism on a scalar]\label{lem:scalar_endomorphism}
    Let $s$ be an expression of type $\mathcal S$ and $f: \mathbb F \to \mathbb F$ a field endomorphism. Suppose $\denot{s}_r$ is well-defined. Then $\denot{\mathsf{apply}(f,s)}_r$ is well-defined, and
    \begin{align}
        \rep^\#(\denot{\mathsf{apply}(f,s)}_r) = f(\rep^\#(\denot{s}_r))
    \end{align}
\end{lemma}
\begin{proof}
    Suppose $\denot{s}_r = [(\phi, W)]$. By definition,
    \begin{align}
        \denot{\mathsf{apply}(f,s)}_r = [(\phi, f \circ W)]
    \end{align}
    We show this is well-defined by showing the value of the expression below is independent of the choice of $\phi$ and $W$: 
    \begin{align}
        \rep^\#(\denot{\mathsf{apply}(f,s)}_r) = \WMC(\phi,f \circ W)
    \end{align}
    It follows from Lemma~\ref{lem:wmc_field_endomorphism} that
    \begin{align}
        \rep^\#(\denot{\mathsf{apply}(f,s)}_r) = f(\WMC(\phi,W)) = f(\rep^\#(\denot{s}_r))
    \end{align}
    This completes the proof.
\end{proof}

\begin{lemma}[Bra and ket]{}
    Let $q \in \mathbb Z_{\geq2}$ and $0 \leq i < q$. Then $\denot{\mathsf{bra}(i,q)}_r$ and $\denot{\mathsf{ket}(i,q)}_r$ are well-defined and
    \begin{align}
        \rep^\#(\denot{\mathsf{bra}(i,q)}_r) &= \bra{i}_q \\
        \rep^\#(\denot{\mathsf{ket}(i,q)}_r) &= \ket{i}_q
    \end{align}
\end{lemma}
\begin{proof}
    We will prove the correctness of the bra. The case of ket is symmetric. The semantics are well-defined by definition.

    We need to show that $\rep^\#(\denot{\mathsf{bra}(i,q)}_r)$ is a row vector with a $1$ at entry $i$ and $0$ everywhere else. This can be done by verifying that, for every $0\leq j < q$, we have
    \begin{align}
        \rep^\#(\denot{\mathsf{bra}(i,q)}_r)\ket{j} = \mathbb 1\{i = j\}
    \end{align}
    Note that we have
    \begin{align}
        \denot{\mathsf{bra}(i,q)}_r = [(x=i,W_1,x,-,q)]
    \end{align}
    with $x$ a $q$-state variable encoding and $W: \var(x) \times \mathbb B \to \mathbb F$ constant $1$. Using the definition of matrix representations, we have
    \begin{align}
        \rep^\#(\denot{\mathsf{bra}(i,q)}_r)\ket{j} &= \WMC(x=i \land x=j, W_1) = \mathbb 1\{i=j\}
    \end{align}
    This proves the lemma.
\end{proof}

\begin{lemma}[Matrix product]{}
    Let $M_1$ and $M_2$ be expressions with types $\mathcal M(q, m \to k)$ and $\mathcal M(q, k \to n)$ respectively. Suppose $\denot{M_1}_r$ and $\denot{M_2}_r$ are well-defined. Then $\denot{M_2\cdot M_1}_r$ is well-defined, and
    \begin{align}
        \rep^\#(\denot{M_2 \cdot M_1}_r) = \rep^\#(\denot{M_2}_r) \cdot \rep^\#(\denot{M_1}_r)
    \end{align}
\end{lemma}
\begin{proof}
    Suppose we have
    \begin{align}
        \denot{M_1}_r &= [(\phi_1, W_1, x, y, q)] \\
        \denot{M_2}_r &= [(\phi_2, W_2, y, z, q)]
    \end{align}
    with $\dom(W_1) \cap \dom(W_2) = \var(y)$. We will show that $\denot{M_2\cdot M_1}_r$ is well-defined by showing that the result of $\rep^\#(\denot{M_2 \cdot M_1}_r)$ is independent of the choice of formulae and weight functions earlier. We can prove the lemma by showing that, for all $0 \leq i < q^n$ and $0 \leq j < q^m$, we have
    \begin{align}
        \bra{i}\rep^\#(\denot{M_2 \cdot M_1}_r)\ket{j} = \bra{i}\rep^\#(\denot{M_2}_r) \cdot \rep^\#(\denot{M_1}_r)\ket{j}
    \end{align}
    By definition, we have
    \begin{align}
        \denot{M_2 \cdot M_1}_r = [\phi_1 \land \phi_2 \land \val_y, W_1 \cdot W_2,x,z,q]
    \end{align}
    Furthermore,
    \begin{align}
        \bra{i}\rep^\#(\denot{M_2 \cdot M_1}_r)\ket{j}
        &= \WMC(\phi_1 \land \phi_2 \land \val_y \land x=j \land z=i, W_1 \cdot W_2)
    \end{align}
    Using the property $\val_v \equiv \bigwedge_{a=0}^{q^k-1} (v=a)$ and Lemma~\ref{lem:wmc_disjoint_formulae}, we have
    \begin{align}
        \bra{i}\rep^\#(\denot{M_2 \cdot M_1}_r)\ket{j}
        &= \sum_{a=0}^{q^k-1} \WMC(\phi_1 \land \phi_2 \land y=a \land x=j \land z=i, W_1 \cdot W_2)
    \end{align}
    Since the only overlap $\phi_1 \land x = j$ and $\phi_2 \land z = i$ have is $\var(y)$, which is the only overlap in the domains of $W_1$ and $W_2$, and all variables in $\var(y)$ are fixed by $y = a$, we can rewrite this as
    \begin{align}
        \bra{i}\rep^\#(\denot{M_2 \cdot M_1}_r)\ket{j}
        &= \sum_{a=0}^{q^k-1} \WMC(\phi_1 \land x = j \land y=a, W_1) \\
        &\quad\ \cdot \WMC(\phi_2 \land z=i \land y=a, W_2) \\
        &= \sum_{a=0}^{q^k-1} \bra{a}\rep^\#(\denot{M_1}_r) \ket{j}\bra{i}\rep^\#(\denot{M_2}_r)\ket{a} \\
        &= \sum_{a=0}^{q^k-1} \bra{i}\rep^\#(\denot{M_2}_r)\ket{a}\bra{a}\rep^\#(\denot{M_1}_r) \ket{j} \\
        &= \bra{i}\rep^\#(\denot{M_2}_r)\cdot \rep^\#(\denot{M_1}_r) \ket{j}
    \end{align}
    This proves the lemma.
\end{proof}


\begin{lemma}[Matrix sum]{}
    Let $M_1$ and $M_2$ be expressions with type $\mathcal M(q,m \to n)$. Suppose $\denot{M_1}_r$ and $\denot{M_2}_r$ are well-defined. Then $\denot{M_1 + M_2}_r$ is also well-defined, and
    \begin{align}
        \rep^\#(\denot{M_1 + M_2}_r) = \rep^\#(\denot{M_1}_r) + \rep^\#(\denot{M_2}_r)
    \end{align}
\end{lemma}
\begin{proof}
    Suppose that
    \begin{align}
        \denot{M_1}_r &= [(\phi_1, W_1, x_1, y_1, q)] \\
        \denot{M_2}_r &= [(\phi_2, W_2, x_2, y_2, q)]
    \end{align}
    with the domains of $W_1$ and $W_2$ disjoint, $\WMC(\top, W_1) \neq 0$, and $\WMC(\top, W_2) \neq 0$. Let $c$ be a Boolean variable and $x$ and $y$ strings of variable encodings, of the same lengths as $x_1$ (or $x_2$) and $y_1$ (or $y_2$), respectively. Let these be chosen in such a way that neither $c$ nor the variables in $x$ and $y$ are contained in either of the domains of $W_1$ and $W_2$. We have defined
    \begin{align}
        \denot{M_1 + M_2}_r &= [(\phi, W_1 \cup W_2 \cup W_c \cup W_{xy}, x, y, q)]
    \end{align}
    with
    \begin{gather}
        \begin{aligned}
            \phi \equiv\ &(\overline c \implies ((x \iff x_1) \land (y \iff y_1) \land \phi_1)) \\
            &\land (c \implies ((x \iff x_2) \land (y \iff y_2) \land \phi_2))
        \end{aligned}
    \end{gather}
    and $W_c: \{c\} \times \mathbb B \to \mathbb F$ and $W_{xy}: (\var(x) \cup \var(y)) \times \mathbb B \to \mathbb F$ defined by $W_c(c) = 1/\WMC(\top,W_1)$, $W_c(\overline c) = 1/\WMC(\top,W_2)$, and $W_{xy}$ constant $1$.

    We will show that $\denot{M_1 + M_2}_r$ is well-defined by showing $\rep^\#(\denot{M_1 + M_2}_r)$ yields the same value, independent of the choice of representations at the start of the proof. Let $0 \leq i < q^n$ and $0 \leq j < q^m$. Write $W = W_1 \cup W_2 \cup W_c \cup W_{xy}$ and $\psi \equiv \phi \land x=j \land y=i$. We have
    \begin{align}
        \bra{i}\rep^\#(\denot{M_1 + M_2}_r)\ket{j}
        &= \WMC(\psi, W)
    \end{align}
    We will show that $\WMC(\psi \land \overline c, W) = \bra{i}\rep^\#(\denot{M_1}_r)\ket{j}$, the case $\WMC(\psi \land c, W)$ is symmetric. Note that
    \begin{align}
        \psi \land \overline c \equiv (x \iff x_1) \land (y \iff y_1) \land \phi_1 \land x=j \land y=i \land \overline c
    \end{align}
    Since $\phi_1$ does not contain the variable $c$, Lemma~\ref{lem:wmc_product_indep} gives
    \begin{align}
        \WMC(\psi \land \overline c, W)
        &= \WMC(\psi, W_1 \cup W_{xy}) \cdot \WMC(\top, W_2) \cdot \WMC(\overline c, W_c) \\
        &= \WMC(\psi, W_1 \cup W_{xy}) \cdot \WMC(\top, W_2) \cdot \frac1{\WMC(\top, W_2)} \\
        &= \WMC(\psi, W_1 \cup W_{xy})
    \end{align}
    We can rewrite $\psi$ to
    \begin{align}
        \psi \equiv \phi_1 \land x = j \land x_1 = j \land y = i \land y_1 = i
    \end{align}
    Since $\psi_1$ only contains variables in the domain of $W_1$, we can rewrite further to
    \begin{align}
        \WMC(\psi \land \overline c, W)
        &= \WMC(\phi_1 \land x_1 = j \land y_1 = i, W_1) \cdot \WMC(x = j \land y = i, W_{xy})
    \end{align}
    Note that there is exactly one satisfying assignment of $x = j \land y = i$, which means the term on the right is $1$, which means
    \begin{align}
        \WMC(\psi \land \overline c, W)
        &= \WMC(\phi \land x_1 = j \land y_1 = i, W_1) \\
        &= \bra{i}\rep^\#[(\phi, W, x_1, y_1, q)] \ket{j} \\
        &= \bra{i}\rep^\#(\denot{M_1}_r)\ket{j}
    \end{align}
    Similarly, it can be proven that
    \begin{align}
        \WMC(\psi \land c, W) = \bra{i}\rep^\#(\denot{M_2}_r)\ket{j}
    \end{align}
    Combining these with Lemma~\ref{lem:wmc_fix_var} gives
    \begin{align}
        \bra{i}\rep^\#(\denot{M_1 + M_2}_r)\ket{j}
        &= \WMC(\psi, W) \\
        &= \WMC(\psi \land c, W) + \WMC(\psi \land \overline c) \\
        &= \bra{i}\rep^\#(\denot{M_1}_r)\ket{j} + \bra{i}\rep^\#(\denot{M_2}_r)\ket{j} \\
        &= \bra{i}(\rep^\#(\denot{M_1}_r) + \rep^\#(\denot{M_2}_r))\ket{j}
    \end{align}
    This proves the lemma.
\end{proof}

\begin{lemma}[Kronecker product]{}
    Let $M_1$ and $M_2$ be expressions with types $\mathcal M(q,m_1 \to n_1)$ and $\mathcal M(q,m_2 \to m_1)$ respectively. Suppose $\denot{M_1}_r$ and $\denot{M_2}_r$ are well-defined. Then $\denot{M_1 \otimes M_2}_r$ is well-defined, and
    \begin{align}
        \rep^\#(\denot{M_1 \otimes M_2}_r) = \rep^\#(\denot{M_1}_r) \otimes \rep^\#(\denot{M_2}_r)
    \end{align}
\end{lemma}
\begin{proof}
    Suppose we have
    \begin{align}
        \denot{M_1}_r &= [(\phi_1, W_1, x_1, y_1, q)] \\
        \denot{M_2}_r &= [(\phi_2, W_2, x_2, y_2, q)]
    \end{align}
    with $\dom(W_1) \cap \dom(W_2) = \varnothing$. Then, by definition,
    \begin{align}
        \denot{M_1 \otimes M_2}_r = [(\phi_1 \land \phi_2, W_1 \cup W_2, x_1x_2, y_1y_2, q)]
    \end{align}
    We show that this is well-defined by showing that the value of $\rep^\#(\denot{M_1 \otimes M_2}_r)$ is independent of the choice of the representations at the start of the proof. We need to show that, for all $0 \leq i_1 < q^{n_1}$, $0 \leq i_2 < q^{n_2}$, $0 \leq j_1 < q^{m_1}$, and $0 \leq j_2 < q^{m_2}$:
    \begin{align}
        \bra{i_1i_2}\rep^\#(\denot{M_1 \otimes M_2}_r)\ket{j_1j_2} &= \bra{i_1i_2}(\rep^\#(\denot{M_1}_r) \otimes \rep^\#(\denot{M_2}_r))\ket{j_1j_2}
    \end{align}
    We have
    \begin{align}
        &\quad\ \bra{i_1i_2}\rep^\#(\denot{M_1 \otimes M_2}_r)\ket{j_1j_2} \\
        &= \WMC(\phi_1 \land \phi_2 \land x_1x_2 = j_1j_2 \land y_1y_2 = i_1i_2, W_1 \cup W_2) \\
        &= \WMC(\phi_1 \land \phi_2 \land x_1 = j_1 \land x_2 = j_2 \land y_1 = i_1 \land y_2 = i_2, W_1 \cup W_2)
    \end{align}
    Lemma~\ref{lem:wmc_product_indep} gives
    \begin{align}
        &\quad\ \bra{i_1i_2}\rep^\#(\denot{M_1 \otimes M_2}_r)\ket{j_1j_2} \\
        &= \WMC(\phi_1 \land x_1 = j_1 \land y_1 = i_1,W_1) \WMC(\phi_2 \land x_2 = j_2 \land y_2 = i_2, W_2) \\
        &= \bra{i_1}\rep^\#(\denot{M_1}_r)\ket{j_1}\bra{i_2}\rep^\#(\denot{M_2}_r)\ket{j_2} \\
        &= \bra{i_1i_2}(\rep^\#(\denot{M_1}_r) \otimes \rep^\#(\denot{M_2}_r))\ket{j_1j_2}
    \end{align}
    This proves the lemma.
\end{proof}

\begin{lemma}[Matrix-scalar multiplication]{}
    Let $s$ be an expression of type $\mathcal S$ and $M$ an expression of type $\mathcal M(q,m\to n)$. Suppose $\denot{s}_r$ and $\denot{M}_r$ are well-defined. Then $\denot{s \cdot M}_r$ is also well-defoned, and
    \begin{align}
        \rep^\#(\denot{s \cdot M}_r) = \rep^\#(\denot{s}_r) \cdot \rep^\#(\denot{M}_r)
    \end{align}
\end{lemma}
\begin{proof}
    We prove that $\denot{s \cdot M}_r$ is well-defined by showing the result of $\rep^\#(\denot{s \cdot M}_r)$ is independent of the choices of representations
    \begin{align}
        \denot{s}_r &= [(\phi_s, W_s)] \\
        \denot{M}_r &= [(\phi, W, x, y, q)]
    \end{align}
    with $\dom(W) \cap \dom(W_s) = \varnothing$. For every $0 \leq i < q^n$ and $0 \leq j < q^m$, we have
    \begin{align}
        \bra{i} \rep^\#(\denot{s \cdot M}_r)\ket{j}
        &= \bra{i}\rep^\#[(\phi \land \phi_s, W \cup W_s, x, y, q)]\ket{j} \\
        &= \WMC(\phi \land \phi_s \land x=j \land y=i, W \cup W_s)
    \end{align}
    Lemma~\ref{lem:wmc_product_indep} gives
    \begin{align}
        \bra{i} \rep^\#(\denot{s \cdot M}_r)\ket{j}
        &= \WMC(\phi_s, W_s)\cdot \WMC(\phi \land x=j \land y = i, W) \\
        &= \rep^\#(\denot{s}_r) \cdot \bra{i}\rep^\#(\denot{M}_r)\ket{j}
    \end{align}
    Since this is true for any $i$ and $j$, this proves the lemma.
\end{proof}

\begin{lemma}[Matrix transpose]{}
    Let $M$ be an expression of type $\mathcal M(q, m \to n)$ and suppose $\denot{M}_r$ is well-defined. Then $\denot{\mathsf{trans}(M)}_r$ is well-defined, and
    \begin{align}
        \rep^\#(\denot{\mathsf{trans}(M)}_r) = (\rep^\#(\denot{M}_r))^T
    \end{align}
\end{lemma}
\begin{proof}
    Suppose $\denot{M}_r = [(\phi, W, x, y, q)]$. By definition,
    \begin{align}
        \denot{\mathsf{trans}(M)}_r = [(\phi, W, y, x, q)]
    \end{align}
    We show that this is well-defined by showing $\rep^\#(\denot{\mathsf{trans}(M)}_r)$ yields the same value, independent of the choice of the representation before.
    
    We want to show that $\bra{i}\rep^\#(\denot{\mathsf{trans}(M)}_r)\ket{j} = \bra{j}\rep^\#(\denot{M}_r))\ket{i}$. Using the formula above, we get
    \begin{align}
        \bra{i}\rep^\#(\denot{\mathsf{trans}(M)}_r)\ket{j}
        &= \WMC(\phi \land y = j \land x = i, W) \\
        &= \WMC(\phi \land x = i \land y = j, W) \\
        &= \bra{j}\rep^\#(\denot{M}_r)\ket{i}
    \end{align}
    This proves the lemma.
\end{proof}

\begin{lemma}[Field endomorphism on a matrix]{}
    Let $M$ be an expression of type $\mathcal M(q, m \to n)$ and $f: \mathbb F \to \mathbb F$ a field endomorphism. Suppose $\denot{M}_r$ is well-defined. Then $\denot{\mathsf{apply}(f,M)}_r$ is well-defined, and
    \begin{align}
        \rep^\#(\denot{\mathsf{apply}(f,M)}_r) = f(\rep^\#(\denot{M}_r))
    \end{align}
    We interpret $f$ being applied to a matrix as it being applied to every entry in the matrix, as it is done for the value semantics $\denot{\cdot}_v$.
\end{lemma}
\begin{proof}
    For $\denot{M}_r = [(\phi, W, x, y, q)]$ we have defined
    \begin{align}
        \denot{\mathsf{apply}(f,M)}_r = [(\phi, f \circ W, x, y, q)]
    \end{align}
    We prove this is well-defined by showing that $\rep^\#(\denot{\mathsf{apply}(f,M)}_r)$ has the same value, independent of the choice of the representation at the start of the proof. Using Lemma~\ref{lem:wmc_field_endomorphism}, we get
    \begin{align}
        \bra{i} \rep^\#(\denot{\mathsf{apply}(f,M)}_r)\ket{j}
        &= \WMC(\phi \land x=j \land y=i, f \circ W) \\
        &= f(\WMC(\phi \land x = j \land y = i, W)) \\
        &= f(\bra{i}\rep^\#(\denot{M}_r)\ket{j})
    \end{align}
    Since this is true for any $0 \leq i < q^n$ and $0 \leq j < q^m$, this proves the lemma.
\end{proof}

\begin{lemma}[Trace]{}
    Let $M$ be an expression of type $\mathcal M(q,n\to n)$ and suppose $\denot{M}_r$ is well-defined. Then $\denot{\mathsf{tr}(M)}_r$ is also well-defined, and
    \begin{align}
        \rep^\#(\denot{\mathsf{tr}(M)}_r) = \tr(\rep^\#(\denot{M}_r))
    \end{align}
\end{lemma}
\begin{proof}
    Suppose that $\denot{M}_r = [(\phi,W,x,y,q)]$, then
    \begin{align}
        \denot{\mathsf{tr}(M)}_r = [(\phi \land (x \iff y) \land \val_x, W)]
    \end{align}
    We prove that this is well-defined by showing the result of $\rep^\#(\denot{\mathsf{tr}(M)}_r)$ is independent of the choice of the representation $\denot{M}_r$. We have
    \begin{align}
        \rep^\#(\denot{\mathsf{tr}(M)}_r)
        &= \rep^\#[(\phi \land (x \iff y) \land \val_x, W)] \\
        &= \WMC(\phi \land (x \iff y) \land \val_x, W)
    \end{align}
    Using the fact that $\val_x \equiv \bigwedge_{a=0}^{q^n-1} (x=a)$ and Lemma~\ref{lem:wmc_disjoint_formulae}, we get
    \begin{align}
        \rep^\#(\denot{\mathsf{tr}(M)}_r)
        &= \sum_{a=0}^{q^n-1}\WMC(\phi \land (x \iff y) \land x = a, W) \\
        &= \sum_{a=0}^{q^n-1}\WMC(\phi \land x = a \land y = a, W) \\
        &= \sum_{a=0}^{q^n-1} \bra{a}\rep^\#(\denot{M}_r)\ket{a} \\
        &= \tr(\rep^\#(\denot{M}_r))
    \end{align}
    This proves the lemma.
\end{proof}

\begin{lemma}[Matrix entry]{}
    Let $M$ be an expression of type $\mathcal M(q, m \to n)$. Suppose we have indices $i$ and $j$ with $0 \leq i < q^n$ and $0 \leq j < q^m$. Furthermore, suppose $\denot{M}_r$ is well-defined. Then $\denot{\mathsf{entry}(i,j,M)}_r$ is well-defined, and
    \begin{align}
        \rep^\#(\denot{\mathsf{entry}(i,j,M)}_r) = (\rep^\#(\denot{M}_r))_{ij}
    \end{align}
\end{lemma}
\begin{proof}
    Suppose $\denot{M}_r = [(\phi, W, x, y, q)]$. Then, by definition, we have
    \begin{align}
        \denot{\mathsf{entry}(i,j,M)}_r = [(\phi \land x = j \land y = i, W)]
    \end{align}
    We prove this is well-defined by showing that $\rep^\#(\denot{\mathsf{entry}(i,j,M)}_r)$ yields the same result, independent of the choice of the representation before. From the definition of matrix representations, we get
    \begin{align}
        (\rep^\#(\denot{M}_r))
        &= \WMC(\phi \land x = j \land y = i, W) \\
        &= \rep^\#[(\phi \land x = j \land y = i, W)] \\
        &= \denot{\mathsf{entry}(i,j,M)}_r
    \end{align}
    This proves the lemma.
\end{proof}

\newpage
\addcontentsline{toc}{section}{References}
\bibliographystyle{plain}
\bibliography{src/bibliography}

\end{document}